\documentclass[runningheads]{llncs}
\usepackage{graphicx}
\usepackage{xspace}
\usepackage{thmtools} 
\usepackage{thm-restate}
\usepackage{todonotes}
\usepackage{enumerate}
\usepackage{soul}
\usepackage{hyperref}
\usepackage{xcolor}

\usepackage{complexity,mathtools}

\renewcommand{\orcidID}[1]{\href{https://orcid.org/#1}{\includegraphics[scale=.03]{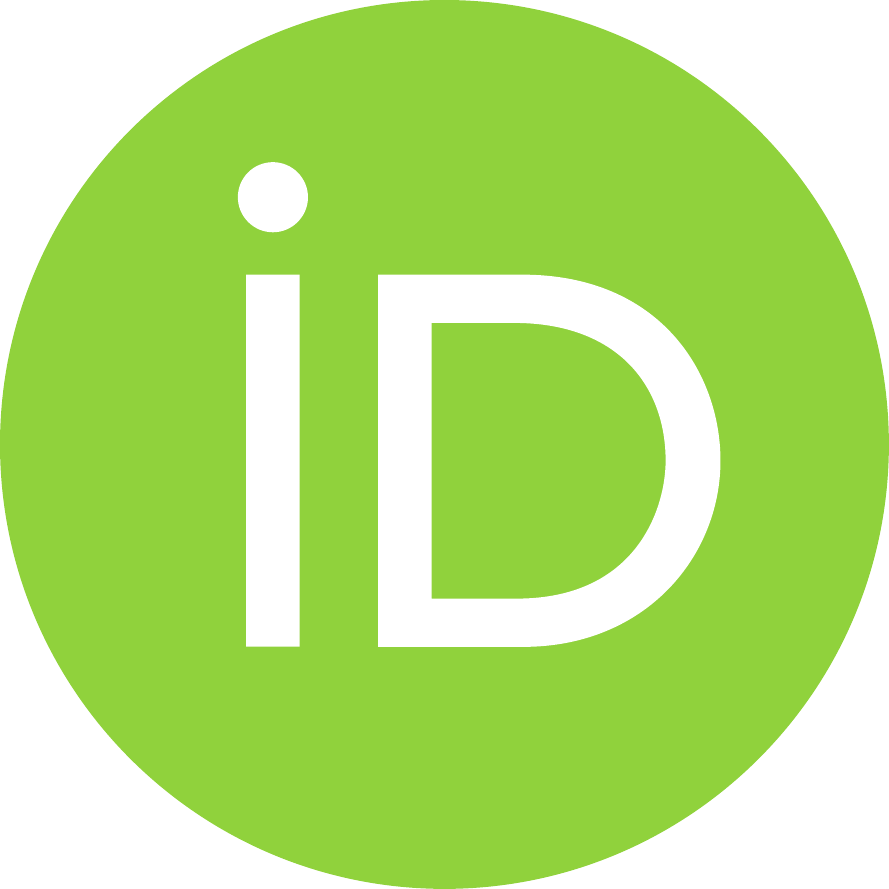}}} 

\usepackage[capitalise,nameinlink]{cleveref}
\Crefname{observation}{Observation}{Observations}
\Crefname{proposition}{Proposition}{Propositions}
\Crefname{claim}{Claim}{Claims}
\Crefname{property}{Property}{Properties}

\Crefname{enumi}{Property}{Properties}
\usepackage{paralist}
\usepackage{subfigure}
\usepackage{amssymb}
\usepackage{cite} 

\hypersetup{
    colorlinks=true,
    linkcolor=blue,
    filecolor=violet,  
    citecolor=violet,
    urlcolor=cyan,
    pdftitle={Overleaf Example},
    pdfpagemode=FullScreen,
}

\newcommand{\ubt}[1]{\ensuremath{\mathrm{UBT(}#1\mathrm{)}}}

\newcommand{\ube}{UBE\xspace}

\definecolor{MyBlue}{rgb}{0.022,0.263,0.394}
\renewcommand{\emph}[1]{{\color{MyBlue}{\em #1}}\xspace}

\newcommand{\shorttitle}{On the Upward Book Thickness Problem\xspace}
\title{%
\shorttitle: Combinatorial and Complexity Results
}
\titlerunning{\shorttitle}

\author{Sujoy Bhore\inst{1}\orcidID{0000-0003-0104-1659} \and
Giordano Da Lozzo\inst{2}\orcidID{0000-0003-2396-5174} \and \\
Fabrizio Montecchiani\inst{3}\orcidID{0000-0002-0543-8912} \and Martin~N\"ollenburg\inst{4}\orcidID{0000-0003-0454-3937}}
\authorrunning{S. Bhore et al.}

\institute{Indian Institute of Science Education and Research, Bhopal, India\\
\email{sujoy.bhore@gmail.com}
\and  Roma Tre University, Rome, Italy\\
\email{giordano.dalozzo@uniroma3.it}
\and Department of Engineering, University of Perugia, Italy\\
\email{fabrizio.montecchiani@unipg.it}
\and Algorithms and Complexity Group, TU Wien, Vienna, Austria\\
\email{noellenburg@ac.tuwien.ac.at }}

\begin{document}

\maketitle

\begin{abstract}

A long-standing conjecture by Heath, Pemmaraju, and Trenk states that the upward book thickness of outerplanar DAGs is bounded above by a constant. In this paper, we show that the conjecture holds for subfamilies of upward outerplanar graphs, namely those whose underlying graph is an internally-triangulated outerpath or a cactus, and those whose biconnected components are $st$-outerplanar graphs. On the complexity side, it is known that deciding whether a graph has upward book thickness $k$ is \NP-hard for any fixed $k \ge 3$. We show that the problem, for any $k \ge 5$, remains \NP-hard for graphs whose domination number is $O(k)$, but it is \textsf{FPT} in the vertex cover number.

\end{abstract}

\section{Introduction}
A \emph{$k$-page book embedding} (or \emph{$k$-stack layout}) of an $n$-vertex graph $G=(V,E)$ is a pair $\langle \pi, \sigma \rangle$ consisting of a bijection $\pi \colon V \rightarrow \{1, \dots, n\}$, defining a total order on $V$, and a page assignment $\sigma \colon E \rightarrow \{1, \dots, k\}$, partitioning $E$ into $k$ subsets $E_i = \{e \in E \mid \sigma(e) = i\}$ ($i=1, \dots, k$) called \emph{pages} (or \emph{stacks}) such that no two edges $uv, wx \in E$ mapped to the same page $\sigma(uv) = \sigma(wx)$ cross in the following sense. 
Assume, w.l.o.g., $\pi(u) < \pi(v)$ and $\pi(w) < \pi(x)$ as well as $\pi(u) < \pi(w)$. 
Then $uv$ and $wx$ \emph{cross} if $\pi(u) < \pi(w) < \pi(v) < \pi(x)$, i.e., their endpoints interleave.
The \emph{book thickness} (or \emph{stack number}) of $G$ is the smallest $k$ for which $G$ admits a $k$-page book embedding. 
Book embeddings and book thickness of graphs are well-studied topics in graph drawing and graph theory~\cite{kainen74,DujmovicW04,ollmann73,ChungLR87}. For instance, it is \NP-complete to decide for $k \ge 2$ if the book thickness of a graph is at most $k$~\cite{BernhartK79,ChungLR87} and it is known that planar graphs have book thickness at most 4~\cite{Yannakakis89}; this bound has recently been shown tight~\cite{KaufmannBKPRU20}. More in general, the book thickness of graphs of genus $g$ is $O(\sqrt{g})$~\cite{DBLP:journals/jal/Malitz94a} and constant upper bounds are known for some families of non-planar graphs~\cite{DBLP:journals/algorithmica/BekosBKR17,DBLP:journals/dcg/DujmovicW07,DBLP:conf/compgeom/BekosLGGMR20}. %

\emph{Upward book embeddings} (UBEs) are a natural extension of book embeddings to directed acyclic graphs (DAGs) with the additional requirement that the vertex order $\pi$ respects the      directions of all edges, i.e., $\pi(u) < \pi(v)$ for each $uv \in E$ (and hence $G$ must be acyclic). Thus the ordering induced by $\pi$ is a topological ordering of $V$.
Book embeddings with different constraints on the vertex ordering have also been studied in~\cite{DBLP:conf/soda/AngeliniLBFP21,DBLP:journals/tcs/AngeliniLN15,DBLP:conf/soda/FulekT20}.
Analogously to book embeddings, the \emph{upward book thickness} (UBT) of a DAG $G$ is defined as the smallest $k$ for which~$G$ admits a $k$-page UBE. 
The notion of upward book embeddings is similar to upward planar drawings~\cite{DBLP:journals/tcs/BattistaT88,Didimo2014}, i.e., crossing-free drawings, where additionally each directed edge $uv$ must be a y-monotone curve from $u$ to $v$. Upward book embeddings have been introduced by Heath et al.~\cite{HPT99b,HPT99a}. They showed that graphs with UBT 1 can be recognized in linear time, whereas Binucci et al.~\cite{DBLP:conf/compgeom/BinucciLGDMP19} proved that deciding the UBT of a graph is generally \NP-complete, even for fixed values of $k\ge 3$. On the positive side, deciding if a graph admits a $2$-page \ube can be solved in polynomial time for $st$-graphs of bounded treewidth~\cite{DBLP:conf/compgeom/BinucciLGDMP19}.

Constant upper bounds on the UBT are known for some graph classes: directed trees have UBT 1~\cite{HPT99a}, unicyclic DAGs,  series-parallel DAGs, and N-free upward planar DAGs have UBT 2~\cite{HPT99a,GGLW06,MS09,AlzohairiR96}. Frati et al.~\cite{FFR13} studied UBEs of upward planar triangulations and gave several conditions under which they have constant UBT.
Interestingly, upward planarity is a necessary condition to obtain constant UBT, as there is a family of planar, but non-upward planar,  DAGs that require $\Omega(n)$ pages in any UBE~\cite{HP97}.
Back in 1999, Heath et al.~\cite{HPT99a} conjectured that the UBT of outerplanar graphs is bounded by a constant, regardless of their upward planarity. Another long-standing open problem~\cite{NP89} is whether upward planar DAGs have constant UBT; in this respect, examples with a lower bound of 5 pages are known~\cite{M20} and \mbox{there is no known upper bound better than $O(n)$.}

\paragraph{Contributions.} In this paper, we contribute to the research on the upward book thickness problem from two different directions. We first report some notable progress towards the conjecture of Heath et al.~\cite{HPT99a}. We consider subfamilies of upward outerplanar graphs (see \cref{ref:prelims} for definitions), namely those whose underlying graph is an internally-triangulated outerpath or a cactus, and those whose biconnected components are $st$-outerplanar graphs, and provide constant upper bounds on their UBT (\cref{se:outerplanar}). Our proofs are constructive and give rise to polynomial-time book embedding algorithms. 
We then investigate the complexity of the problem (\cref{se:complexity}) and show that for any $k \ge 5$ it remains \NP-complete for graphs whose domination number is in $O(k)$. On the positive side, we prove that the upward book thickness problem is fixed-parameter tractable in the vertex cover number. These two results narrow the gap between tractable and intractable parameterizations of the problem. 
Proofs of  statements marked with a $(\star)$ have been sketched or omitted~and~are~in~the~appendix.
\clearpage

\section{Preliminaries}\label{ref:prelims}
We assume familiarity with basic concepts in graph drawing (see also~\cite{DBLP:books/ph/BattistaETT99}). 

\paragraph{BC-tree.} The \emph{BC-tree} of a connected graph $G$ is the incidence graph between the  (maximal) biconnected components of $G$, called \emph{blocks}, and the cut-vertices of $G$. A block is \emph{trivial} if it consists of a single edge, otherwise it is \emph{non-trivial}.

\paragraph{Outerplanarity.} 
An \emph{outerplanar graph} is a graph that admits an \emph{outerplanar drawing}, i.e., a planar drawing in which all vertices are on the outer face, which defines an \emph{outerplanar embedding}. 
Unless otherwise specified, we will assume our graphs to have planar or outerplanar embeddings. An outerplanar graph $G$ is \emph{internally triangulated} if it is biconnected and all its inner faces are cycles of length 3. An edge of $G$ is \emph{outer} if it belongs to the outer face of $G$, and it is \emph{inner} otherwise. A \emph{cactus} is a connected outerplanar graph in which any two simple cycles have at most one vertex in common. Therefore, the blocks of a cactus graph are either single edges (and hence trivial) or cycles. The \emph{weak dual} $\overline{G}$ of a planar graph~$G$ is the graph having a node for each inner face of $G$, and an edge between two nodes if and only if the two corresponding faces share an edge. For an outerplanar graph~$G$, its weak dual $\overline{G}$ is a tree. If $\overline{G}$ is a path, then $G$ is an \emph{outerpath}. A \emph{fan} is an internally-triangulated outerpath whose inner edges all share an end-vertex.

\paragraph{Directed graphs.} A \emph{directed graph} $G=(V,E)$, or \emph{digraph}, is a graph whose edges have an orientation.  We assume each edge $e=uv$ of $G$ to be oriented from $u$ to $v$, and hence denote $u$ and $v$ as the \emph{tail} and \emph{head} of $e$, respectively. A vertex $u$ of $G$ is a \emph{source} (resp.\ a \emph{sink}) if it is the tail (resp.\ the head) of all its incident edges. If $u$ is neither a source nor a sink of $G$, then it is  \emph{internal}. 
A  \emph{DAG} is a digraph that contains no directed cycle. An \emph{$st$-DAG} is a DAG with a single source $s$ and a single sink $t$; if needed, we may use different letters to denote $s$ and $t$. A digraph is \emph{upward (outer)planar}, if it has a (outer)planar drawing such that each edge is a y-monotone curve. Such a drawing (if any) defines an upward (outer)planar embedding. An upward planar digraph $G$ (with an upward planar embedding) is always a DAG and it is \emph{bimodal}, that is, the sets of incoming and outgoing edges at each vertex $v$ of $G$ are contiguous around $v$ (see also~\cite{DBLP:books/ph/BattistaETT99}).  
The \emph{underlying graph} of a digraph is the graph obtained by disregarding the edge orientations. An \emph{$st$-outerplanar graph} (resp.\ \emph{$st$-outerpath}) is an $st$-DAG whose underlying graph is outerplanar graph (resp. \ an outerpath). An \emph{$st$-fan} is an $st$-DAG whose  underlying graph is a fan and whose inner edges have $s$~as~an~end-vertex.

\begin{restatable}[$\star$]{lemma}{lembimodality}\label{lem:bimodality}Let $G$ be an upward outerplanar graph and let~$c$ be a cut-vertex of~$G$. Then there are at most two blocks of $G$ for which $c$ is internal.
\end{restatable}

\newcommand{\lembimodalityproof}{
\begin{proof}
Observe that if $c$ is internal for a block $\beta$ of $G$, then $\beta$ must be non-trivial.
Suppose, for a contradiction, that there exist three non-trivial blocks $\beta_1$, $\beta_2$, and $\beta_3$ of $G$ incident to $c$ for which $c$ is internal.
Let $\cal E$ be the upward outerplanar embedding of $G$. 
Since each $\beta_i$ is non-trivial, it contains a cycle of edges $C_i$, for which $c$ is internal.
Let $e_i$ and $h_i$ be the edges of $C_i$ for which $c$ is the head and the tail, respectively. Assume, w.l.o.g., that $e_1$, $e_2$, and $e_3$ appear in this left-to-right order in $\cal E$. Then, by the upward planarity of $\cal E$, it holds that the edges $h_1$, $h_2$, and $h_3$ appear in this left-to-right order in $\cal E$. Thus, by planarity, $C_2$ must enclose in its interior all the vertices of either $C_1$ or $C_3$ (except for $c$), contradicting the fact that $\cal E$ is outerplanar.
\end{proof}}

\paragraph{Basic operations.}
Let $\pi$ and $\pi'$ be two orderings over vertex sets $V$ and  $V' \subseteq V$, respectively. Then  $\pi$ \emph{extends}  $\pi'$ if for any two vertices $u,v \in V'$ with $\pi'(u) < \pi'(v)$, it holds $\pi(u) < \pi(v)$. We may denote an ordering $\pi$  as a list $\langle v_1,v_2,\dots,v_{|V|} \rangle$ and use the \emph{concatenation operator} $\circ$ to define an ordering from other lists, e.g., we may obtain the ordering $\pi=\langle v_1,v_2,v_3,v_4 \rangle$ as $\pi = \pi_1 \circ \pi_2$, where $\pi_1=\langle v_1,v_2 \rangle$ and $\pi_2=\langle v_3,v_4 \rangle$. Also, let $\pi$ be an ordering over $V$ and let $u \in V$. We denote by \emph{$\pi_{u^-}$} and \emph{$\pi_{u^+}$} the two orderings such that $\pi=\pi_{u^-} \circ \langle u\rangle \circ \pi_{u^+}$. Consider two orderings $\pi$ over $V$ and $\pi'$ over $V'$ with $V \cap V'=\{u,v\}$, and such that: (i) $u$ and $v$ are consecutive in $\pi$, and (ii) $u$ and $v$ are the first and the last vertex of $\pi'$, respectively. The ordering $\pi^*$ over $V \cup V'$ obtained by \emph{merging} $\pi$ and $\pi'$ is $\pi^*=\pi_{u^-} \circ \pi' \circ \pi_{v^+}$. Note that $\pi^*$ extends both $\pi$ and $\pi'$.

\section{Book thickness of outerplanar graphs}\label{se:outerplanar}

In this section, we study the UBT of three families of upward outerplanar graphs. We begin with internally-triangulated upward outerpaths (\cref{sse:max-outerpaths}), which are biconnected and may have multiple sources and sinks. We then continue with families of outerplanar graphs that are not biconnected but whose biconnected components have a simple structure, namely outerplanar graphs whose biconnected components are $st$-DAGs (\cref{sse:outerplanar}),  and cactus graphs (\cref{sse:cactus}).

\subsection{Internally-triangulated upward outerpaths}\label{sse:max-outerpaths}

In this subsection, we assume our graphs to be internally triangulated.
We will exploit the following definition and lemmas for our constructions. 

\begin{definition}\label{def:consecutive}
Let $G$ be an $st$-outerpath and $uv$ be an outer edge different from~$st$. An  \ube $\langle \pi, \sigma \rangle$ of $G$ is  \emph{$uv$-consecutive} if the following properties hold:
\begin{inparaenum}[\bf (i)]
    \item $u$ and $v$ are consecutive in $\pi$,
    \item the edges incident to $s$ lie on one page, and 
    \item the edges incident to $t$ lie on at most two pages.
\end{inparaenum}
\end{definition}

An $st$-outerplanar graph is \emph{one-sided} if the edge $st$ is an outer edge.

\begin{restatable}{lemma}{lemStFanDrawing}\label{basic:st-fan-drawing}Let $G$ be a one-sided $st$-outerplanar graph. Then, $G$ admits a $1$-page \ube $\langle \pi, \sigma \rangle$ that is $uv$-consecutive for each outer edge $uv \neq st$. 
\end{restatable}

\newcommand{\lemStFanDrawingProof}{\begin{proof}Consider the path of the outer face of $G$ that encompasses all vertices of $G$ and does not contain the edge $st$. Let $\langle \pi, \sigma \rangle$ be the $1$-page \ube in which $\pi$ is the ordering defined by such path. One easily verifies that no two edges cross and that the endpoints $u, v$ of each outer edge $uv \neq st$ are consecutive in $\pi$.\end{proof}}

\lemStFanDrawingProof

\begin{lemma}\label{lem:s-fan}
Let $G$ be an $st$-fan and let $uv \neq st$ be an outer edge of $G$. Then, $G$ admits a $2$-page $uv$-consecutive \ube.
\end{lemma}

\begin{proof}
\begin{figure}[t]
    \centering
    \subfigure[]{
    \includegraphics[page=4]{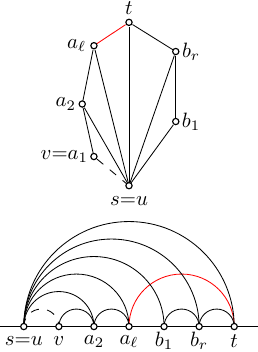}\label{fig:fan-IncidentToS}
    }
    \subfigure[]{
    \includegraphics[page=3]{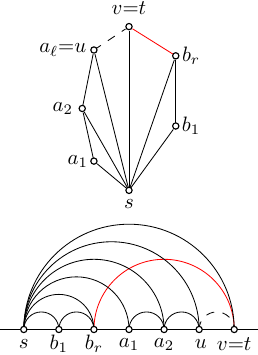}\label{fig:fan-IncidentToT}
    }
    \caption{Cases for the proof of \cref{lem:s-fan}. The edge $uv$ is dashed. In all figures, the drawings are upward, hence the edge orientations are implied.}
    \label{fig:fan}
\end{figure}
Let $P_\ell=\{s,a_1,\dots,a_\ell,t\}$ and $P_r=\{s,b_1,\dots,b_r,t\}$ be the left and right $st$-paths of the outer face of $G$, respectively. Then the edge $uv$ belongs to either $P_\ell$ or $P_r$. We show how to construct an \ube $\langle \pi, \sigma \rangle$ of $G$ that satisfies the requirements of the lemma, when $uv$ belongs to $P_\ell$ (see \cref{fig:fan}); the construction when $uv$ belongs to $P_r$ is symmetric (it suffices to flip the embedding along $st$).

Since $uv \neq st$, we have either (a) $u = s$ and $v \neq t$, or (b) $v = t$ and $u \neq s$, or (c) $\{u,v\} \cap \{s,t\} = \emptyset$.
In case (a), refer to \cref{fig:fan-IncidentToS}. We set $\pi=\langle s,a_1,\dots,a_\ell,b_1,\dots,b_r,t\rangle$ (that is, we place $P_\ell$ before $P_r$), $\sigma(e)=1$ for each edge $e \neq a_\ell t$, and $\sigma(a_\ell t) = 2$.
In case (b), refer to \cref{fig:fan-IncidentToT}. We set $\pi=\langle s,b_1,\dots,b_r,a_1,\dots,a_\ell,t\rangle$ (that is, we place $P_\ell$ after~$P_r$), $\sigma(e)=1$ for each edge $e \neq b_r t$, and $\sigma(b_r t) = 2$. 
In case (c), we can set $\pi$ and $\sigma$ as in any of case (a) and (b).
In all three cases, all outer edges (including $uv$) are consecutive, except for one edge $e$ incident to $t$; also, all edges (including those incident to $s$) are assigned to the same page, except for $e$, which is assigned to a second page.%
\end{proof}

\newcommand{\good}{primary\xspace}
\newcommand{\app}{appendage\xspace}

\noindent The next definition allows us to split an $st$-outerpath into two simpler~graphs. The \emph{extreme faces} of an $st$-outerpath $G$ are the two faces that correspond to the vertices of $\overline{G}$ having degree one.

\begin{definition}\label{def:good}
An $st$-outerpath $G$ is \emph{\good} if and only if the path forming $\overline{G}$ has one extreme face incident to $s$.
\end{definition}

Let $G'$ be an $st$-outerpath and refer to \cref{fig:good}. Consider the subgraph $F_s$ of $G'$ induced by $s$ and its neighbors, note that this is an $sw$-fan. Assuming $F_s \neq G'$,  and since $G'$ is an outerpath, one or two edges on the outer face of $F_s$ are separation pairs for $G'$. In the former case, it follows that $G'$ is \good, since $\overline{G'}$ is a path having one face of $F_s$ as its extreme face. In the latter case, since $G'$ has a single source $s$ and a single sink $t$, (at least) one of the two separation pairs splits $G'$ into a one-sided $uv$-outerpath $H_1$ (for some vertices $u,v$ of $F_s$) and into a \good $st$-outerpath $G$. We will call $H_1$, the \emph{\app at $uv$} of $G'$. %

\begin{figure}[tb!]
    \centering
    \subfigure[]{
    \includegraphics[page=1]{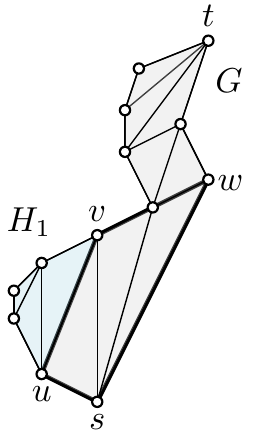}\label{fig:good}
    }\hfil
     \subfigure[]{
    \includegraphics[page=2]{figs/nice-new.pdf}\label{fig:stfandec}
    }\hfil
     \subfigure[]{
    \includegraphics[page=3]{figs/nice-new.pdf}\label{fig:stgoodouter}
    }
    \caption{(a) Decomposing $G'$ into an \app $H_1$ (light blue) and a \good $st$-outerpath $G$ (light gray). (b) An $st$-fan decomposition of $G$ with differently colored fans, and fat edges $e_i$.  (c) A \ube of $G$ for the proof of \cref{lem:good-st-outerpath}.}
\end{figure}

Let $G$ be an $st$-outerpath (not necessarily \good). Consider a subgraph $F$ of $G$ that is an $xy$-fan (for some vertices $x,y$ of $G$). Let $\langle f_1,\dots,f_h \rangle$ be the ordered list of faces forming the path $\overline{G}$. Note that $F$ is the subgraph of $G$ formed by a subset of  faces that are consecutive in the path $\langle f_1,\dots,f_h \rangle$. Let $f_i$ be the face of $F$ with the highest index. We say that $F$ is \emph{incrementally maximal} if $i=h$ or $F \cup f_{i+1}$ is not an $xy$-fan. 
We state another key definition; refer to \cref{fig:stfandec}.
\begin{definition}\label{def:st-fandec}An \emph{$st$-fan decomposition} of an $st$-outerpath $G$ is a sequence of $s_i t_i$-fans $F_i \subseteq G$, with $i=1,\dots,k$, such that: 
(i) $F_i$ is incrementally maximal;  
(ii) For any $1 \le i < j \le k$, $F_i$ and $F_j$ do not share any edge if $j>i+1$, while $F_i$ and $F_{i+1}$ share a single edge, which we denote by $e_i$; 
(iii) $s_1=s$; 
(iv) the tail of $e_i$ is $s_{i+1}$; 
(v) edge $e_i \neq s_it_i$; and
(vi) $\bigcup_{i=1}^k F_i=G$.
\end{definition}

We next show that \good $st$-outerpaths always have $st$-fan decompositions.

\begin{restatable}[$\star$]{lemma}{fandecunique}\label{basic:fandec-unique}
Every \good $st$-outerpath $G$ admits an $st$-fan decomposition.
\end{restatable}
\begin{proof}[Sketch]
By \cref{def:good}, one extreme face of $\overline{G}$ is incident to $s$; we denote such face by $f_1$, and the other extreme face of $\overline{G}$ by $f_h$. We construct the $st$-fan decomposition as follows. We initialize $F_1=f_1$ and we parse the faces of $G$ in the order defined by $\overline{G}$ from $f_1$ to $f_h$. Let $F_i$ be the current $s_it_i$-fan and let $f_j$ be the last visited face (for some $1 < j < h$). If $F_i \cup f_{j+1}$ is an $s_i t_i$-fan, we set $F_i = F_i \cup f_{j+1}$, otherwise we finalize $F_i$ and initialize $F_{i+1}=f_j$. 
\end{proof}
\newcommand{\fandecuniqueProof}{
\begin{proof}
By \cref{def:good}, one extreme face of $\overline{G}$ is incident to $s$; we denote such face by $f_1$, and the other extreme face of $\overline{G}$ by $f_h$. We construct the $st$-fan decomposition as follows. We initialize $F_1=f_1$ and we parse the faces of $G$ in the order defined by $\overline{G}$ from $f_1$ to $f_h$. Let $F_i$ be the current $s_it_i$-fan and let $f_j$ be the last visited face (for some $1 < j < h$). If $F_i \cup f_{j+1}$ is an $s_i t_i$-fan, we set $F_i = F_i \cup f_{j+1}$, otherwise we finalize $F_i$ and set $F_{i+1}=f_j$. 

We now prove that the computed set of fans, denoted by $F_1,F_2,\dots,F_k$, is indeed an $st$-fan decomposition of $G$. If $k=1$, one easily verifies that all conditions of \cref{def:st-fandec} are satisfied. Suppose now that $k>1$. To prove condition (i), observe that for a fan $F_i$ to be not incrementally maximal, $F_{i+1}$ (if $i<k$) must contain a face $f$ such that $F_i \cup f$ is an $s_i t_i$-fan. However, $F_{i+1}$ cannot contain $f$ by construction, since we iteratively expanded $F_i$ towards $f_h$ until it was no further possible. Condition (ii) is satisfied because we parse the faces of $G$ one by one and never assign the same face to two fans. Condition (iii) is satisfied because the first face $f_1$ considered is incident to $s$. Condition (iv) is satisfied because $F_i$ and $F_{i+1}$ share edge $e_i$ and if the tail of $e_i$ were not $s_{i+1}$, then $G$ would contain another source distinct from $s$. Concerning condition (v), let $f$ be the face of $F_{i+1}$ incident to $e_i$. If $e_i=s_it_i$, by (iv) it holds $s_i=s_{i+1}$, and therefore $F_i \cup f$ would still be an $s_it_i$-fan, which is not possible by construction. Condition (vi) is satisfied unless there exists some face $f$ that, in $\overline{G}$, comes before $f_1$ and hence is not parsed by our construction. However this is not possible because $f_1$ and $f_h$ are the two extreme faces of $\overline{G}$.
\end{proof}}

\begin{restatable}[$\star$]{lemma}{lemSharededge}\label{le:sharededge}
Let $G$ be a \good $st$-outerpath and let $F_1,F_2,\dots,F_k$ be an $st$-fan decomposition of $G$. Any two fans $F_i$ and $F_{i+1}$ are such that if $F_{i+1}$ is not one-sided, then $e_i=s_{i+1}t_{i}$.
\end{restatable}

\newcommand{\lemSharededgeProof}{\begin{proof}
Let $G_{j}$ be the subgraph of $G$ induced by $F_1,F_2,\dots,F_{j}$, with $j=1,\dots,k$. 
Since $G$ is an $st$-outerpath and $F_1$ contains $s$ by \cref{def:st-fandec}, vertex $t_{i+1}$ is the sink of $G_{i+1}$. Also, since $F_{i+1}$ is not one-sided, it holds $t_i\neq t_{i+1}$. Hence, $G_{i+1}$ contains a directed path from $t_{i}$ to $t_{i+1}$. It follows that $e_{i}={s_{i+1} t_{i}}$, otherwise $s_{i+1}$ would not belong to $G_i$ and $G$ would contain two sources.
\end{proof}}

\begin{lemma}\label{lem:good-st-outerpath}
Let $G$ be a \good $st$-outerpath and let $F_1,F_2,\dots,F_k$ be an $st$-fan decomposition of $G$. Also, let $e \neq s_kt_k$ be an outer edge of $F_k$. Then, $G$ admits a $4$-page $e$-consecutive \ube $\langle \pi, \sigma \rangle$. 
\end{lemma}

\begin{proof}
We construct a $4$-page $e$-consecutive \ube of $G$ by induction on $k$;  see also \cref{fig:stgoodouter}, which shows a \ube of the \good $st$-outerpath in \cref{fig:stfandec}. 

Suppose $k=1$. Then $G$ consists of the single $s_1 t_1$-fan and a $2$-page $e$-consecutive \ube of $G$ exists by \cref{lem:s-fan}.

Suppose now $k>1$. Let $G_{i}$ be the subgraph of $G$ induced by $F_1 \cup F_2 \cup \dots \cup F_{i}$, for each $1 \le i \le k$. Recall that $e_{i}$ is the  edge shared by $F_{i}$ and $F_{i+1}$ and that the tail of $e_{i}$ coincides with $s_{i+1}$. Let $\langle \pi, \sigma \rangle$ be a $4$-page $e_{k-1}$-consecutive \ube of $G_{k-1}$, which exists by induction since $e_{k-1} \neq s_{k-1}t_{k-1}$  by condition $(v)$ of \cref{def:st-fandec}, and distinguish whether $F_k$ is one-sided or not. 

If $F_k$ is one-sided, it admits a $1$-page $e$-consecutive \ube $\langle \pi',\sigma'\rangle$ by \cref{basic:st-fan-drawing}. In particular, $e \neq e_{k-1}$, since $e_{k-1}$ is not an outer edge of $G_k$. Also, $e_{k-1}=s_k t_k$, and hence the two vertices shared by $\pi$ and $\pi'$ are $s_{k},t_{k}$, which are consecutive in $\pi$ and are the first and the last vertex of $\pi'$. Then we define a $4$-page $e$-consecutive \ube $\langle \pi^*,\sigma^*\rangle$ of $G_{k}$ as follows. The ordering  $\pi^*$ is obtained by merging $\pi'$ and $\pi$. Since $e_{k-1}=s_k t_k$ is uncrossed (over all pages of $\sigma$), for every edge $e$ of $F_k$, we set $\sigma^*(e) = \sigma(e_{k-1})$, while for every other edge $e$ we set $\sigma^*(e) = \sigma(e)$.

If $F_k$ is not one-sided, by \cref{lem:s-fan}, $F_k$ admits a $2$-page $e$-consecutive \ube $\langle \pi',\sigma'\rangle$. Then we obtain a $4$-page $e$-consecutive \ube $\langle \pi^*,\sigma^*\rangle$  of $G_{k}$ as follows.  
By \cref{le:sharededge}, it holds $e_{k-1}={s_k t_{k-1}}$, and thus $s_{k}$ and $t_{k-1}$ are the second-to-last and the last vertex in $\pi$, respectively; also, $s_k$ is the first vertex in $\pi'$.
We set $\pi^* = \pi_{{t_k}^-} \circ \pi'$. 
Since $\langle \pi, \sigma \rangle$ is a $e_{k-1}$-consecutive \ube of $G_{k-1}$, by \cref{def:consecutive} we know that the edges incident to $t_{k-1}$ can use up to two different pages. On the other hand, these are the only edges that can be crossed by an edge of $F_k$ assigned to one of these two pages. Therefore, in \cref{lem:s-fan}, we can assume $\sigma'$ uses the two pages not used by the edges incident to $t_{k-1}$, and set $\sigma^*(e)=\sigma'(e)$ for every edge $e$ of $F_k$ and $\sigma^*(e) = \sigma(e)$ for every other edge. %
\end{proof}

Next we show how to reinsert the \app $H_1$ (\cref{lem:st-outerpath}). To this aim, we first provide a more general tool (\cref{lem:one-insertion}) that will be useful also in \cref{sse:outerplanar}.

\begin{restatable}[$\star$]{lemma}{lemOneIsertion}\label{lem:one-insertion}Let $G=(V,E)$ be a \good $st$-outerpath  with a $4$-page $e$-consecutive \ube of $G$ obtained by using \cref{lem:good-st-outerpath}, for some outer edge $e$ of $G$. For $i=1,\dots,h$, let $H_i=(V_i,E_i)$ be a one-sided $u_iv_i$-outerplanar graph such that $E \cap E_i = \{u_i,v_i\}$. Then $G'=G \cup H_1 \cup \dots \cup H_h$ admits  a $4$-page $e$-consecutive \ube, as long as $e \neq e_i$ ($i=1,\dots,h$). 
\end{restatable}

\newcommand{\lemOneIsertionProof}{\begin{proof}
Let $\langle \pi, \sigma \rangle$ be the $4$-page $e$-consecutive \ube of $G$ obtained by \cref{lem:good-st-outerpath}. In particular, let $F_1,F_2,\dots,F_k$ be the $st$-fan decomposition of $G$ in input to \cref{lem:good-st-outerpath}. We first show how to compute the desired \ube for $G \cup H_1$.

Graph $H_1$ admits a $1$-page \ube $\langle \pi', \sigma' \rangle$ by \cref{basic:st-fan-drawing}. Also, the edge $u_1v_1$ belongs to some fan $F_j$ and, since $u_1v_1$ is shared with $H_1$, it must be $u_1v_1 \neq e_j$ (where $e_j$ is the edge shared by $F_j$ and $F_{j+1}$). If $u_1$ and $v_1$ are consecutive in $\pi$, then a $4$-page $e$-consecutive \ube $\langle \pi^*, \sigma^* \rangle$ of $G \cup H_1$ can be obtained by merging  $\pi$ and $\pi'$ and by setting $\sigma^*(e) = \sigma(u_1v_1)$ for every edge $e$ of $H_1$, and $\sigma^*(e) = \sigma(e)$ for every other edge. If $u_1$ and $v_1$ are not consecutive in $\pi$, let $P_\ell=\{s_j,a_1,\dots,a_\ell,t_j\}$ and $P_r=\{s_j,b_1,\dots,b_r,t_j\}$ be the $s_jt_j$-paths formed by the outer edges that lie to the left and to the right of $s_jt_j$ in $F_j$, respectively. Recall that $F_j$ has been embedded by using \cref{lem:s-fan} and hence the only outer edges of $F_j$ that are not consecutive are either $s_j b_1$ and $a_\ell t_j$ or $s_j a_1$ and $b_r t_j$. For ease of description, suppose the non-consecutive edges are $s_j b_1$ and $a_\ell t_j$, as otherwise the argument is symmetric. Note that it cannot be $u_1 = s_j$, hence we have $u_1v_1 = a_\ell t_j$. Then in the proof of \cref{lem:s-fan}, edge $u_1v_1$ is the only edge assigned to the second page of the \ube. Also, by \cref{le:sharededge}, $a_\ell$ and $b_1$ are consecutive in $\pi$, because $F_{j+1}$ (if it exists) is either one-sided or $e_j=b_r t_j$. By exploiting these properties, we construct a $4$-page \ube $\langle \pi^*,\sigma^* \rangle $ of $G \cup H_1$ as follows. We set $\pi^* = \pi_{u_1^-} \circ \pi'_{v_1^-} \circ \pi_{u_1^+}$. Note that, in this way, if an edge $e$ of $H_1$ is crossed by an edge $e'$ of $G$, then $e$ is incident to $t_j$, and consequently $e'$ also crosses $u_1v_1$. Then we set $\sigma'(e)=\sigma(e)$ for each edge $e$ of $G$, and $\sigma'(e)=\sigma(u_1v_1)$ for the remaining edges of $H_1$. This immediately implies that if an edge $e$ of $H_1$ is crossed by an edge $e'$ of $G$, then the two edges are on different pages, otherwise $e'$ and $u_1v_1$ would also cross and belong to the same page, which is not possible. 

Since it holds $u_iv_i \neq u_jv_j$ for any $1 \le i < j \le h$, by repeating the above procedure for each $H_i$, we obtain a $4$-page $e$-consecutive \ube $\langle \pi^*, \sigma^* \rangle$ of $G'$. In particular, if two edges $e$ and $e'$ cross such that $e$ and $e'$ both belong to $G$, then they also cross in $\pi$ because $\pi^*$ extends $\pi$, and hence $\sigma^*(e) \neq \sigma^*(e')$ because $\sigma^*(e)=\sigma(e)$ and $\sigma^*(e')=\sigma(e')$. If $e$ and $e'$ both belong to the same $H_i$, they do not cross because $\pi^*$ extends $\pi'$. If $e$ belongs to $G$ and $e'$ to some $H_i$, then $e$ also crosses $u_i v_i$ by construction of $\pi^*$, and again $\sigma^*(e) \neq \sigma^*(e')$ because $\sigma^*(e) \neq \sigma^*(u_iv_i)$ while $\sigma^*(e') = \sigma^*(u_iv_i)$. If $e$ and $e'$ belong to $H_i$ and $H_j$ respectively, for some $i \neq j$, then again by construction of $\pi^*$ we have that $u_iv_i$ also crosses $u_jv_j$ and $\sigma^*(e) \neq \sigma^*(e')$ because $\sigma^*(e) = \sigma^*(u_iv_i)$, $\sigma^*(e') = \sigma^*(u_jv_j)$, and $\sigma^*(u_iv_i) \neq \sigma^*(u_jv_j)$. 
\end{proof}}

\begin{restatable}[$\star$]{lemma}{lemStouterpath}\label{lem:st-outerpath}
Let $G'=G \cup H_1$ be an $st$-outerpath, such that $G$ is a \good $st$-outerpath and $H_1$ is the \app at $uv$ of $G'$. Let $F_1,F_2,\dots,F_k$ be an $st$-fan decomposition of $G$. Let $e$ be an outer edge of $G'$ that belongs to either $F_k$ or $H_1$, or to $F_1$ if $H_1=\emptyset$. Also, if $e \in F_k$, then $e \neq s_kt_k$, otherwise  \ $e \neq s_1t_1$. Then, $G'$ admits an $e$-consecutive $4$-page  \ube $\langle \pi, \sigma \rangle$.
\end{restatable}
\newcommand{\lemStouterpathProof}{\begin{proof}
If $e$ is an outer edge of $F_k$, since the \app $H_1$ of $G'$ is a one-sided $uv$-outerpath that only shares the edge $uv$ with $G$, the  lemma is an immediate consequence of \cref{lem:one-insertion}.
If $e$ is an outer edge of $F_1$ or $H_1$, we first reverse the orientation of the edges of $G'$ and obtain a $ts$-outerpath $G''$. In a $ts$-fan decomposition of the \good $ts$-outerpath of $G''$, which we denote by $F''_1,F''_2,\dots,F''_k$, the edge $e$ now belongs to $F''_k$ (because it belongs to $H_1$ or to $F_1$ if $H_1=\emptyset$). Thus again we can compute a $4$-page $e$-consecutive \ube $\langle \pi'', \sigma'' \rangle$. By reversing the ordering $\pi''$, we obtain a $4$-page $e$-consecutive \ube of $G'$.
\end{proof}}

\begin{figure}[t]
    \centering
    \includegraphics[page=2]{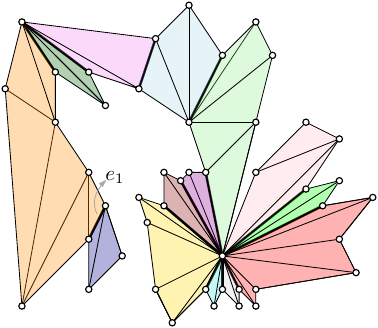}
    \caption{An $st$-outerpath decomposition of an upward outerpath; edges $e_i$ are fat ($e_1$\label{fig:st-outer-dec}}
\end{figure}

We now define a decomposition of an upward outerpath $G$; refer to \cref{fig:st-outer-dec}. Let $P \subset G$ be an $st$-outerpath, then $P$ is the subgraph of $G$ formed by a subset of consecutive faces of $\overline{G}=\langle f_1,\dots,f_h \rangle$. Let $f_j$ and $f_{j'}$ be the faces of $P$ with the smallest and highest index, respectively. Let $F^i$ be the incrementally maximal $s^i t^i$-fan of $P_i$ (assuming $\langle f_j,\dots,f_{j'} \rangle$ to be the ordered list of faces of $\overline{P_i}$). We say that $P$ is \emph{incrementally maximal} if $i=h$ or if $P \cup f_{i+1}$ is not an $st$-outerpath or if $P \cup f_{i+1}$ is still an $st$-outerpath but the edge $s^it^i$ of $F^i$ is an outer edge~of~$P$.

\begin{definition}\label{def:st-outerdec}
An \emph{$st$-outerpath decomposition} of an upward outerpath $G$ is a sequence of $s_i t_i$-outerpaths $P_i \subseteq G$, with $i=1,2,\dots,m$, such that: 
\begin{inparaenum}[(i)]
\item \label{prop:outer-maximality} $P_i$ is incrementally maximal;  
\item \label{prop:outer-share-one-edge} For any $1 \le i < j \le m$, $P_i$ and $P_{j}$ share a single edge if $j=i+1$, which we denote by $e_i$, while they do not share any edge otherwise; and 
\item \label{prop:outer-partition} $\bigcup_{i=1}^m P_i=G$.
\end{inparaenum}
\end{definition}

\begin{restatable}[$\star$]{lemma}{lemStOuterExsists}\label{lem:stouterexsists}Every upward outerpath admits an $st$-outerpath decomposition. 
\end{restatable}

\newcommand{\lemStOuterExsistsProof}{\begin{proof}
Let $G$ be an internally-triangulated upward outerpath.
First, consider the weak dual $\overline{G}=\langle f_1,\dots,f_h \rangle$ of $G$ and the outerpaths $P_1,P_2,\dots,P_m$ constructed as follows. We initialize $P_1=f_1$ and we parse the faces of $G$ in the order defined by $\overline{G}$ from $f_1$ to $f_h$. Let $P_i$ be the current $s_it_i$-outerpath and let $f_j$ be the last visited face (for some $1 < j < h$). We set $P_i = P_i \cup f_j$ if $P_i \cup f_j$ is still a single-source single-sink outerpath, otherwise we set $P_{i+1}=f_j$ and proceed. Clearly, at the end of this process, we have obtained a sequence of $s_i t_i$-outerpaths that satisfies Properties  (\ref{prop:outer-share-one-edge}) and (\ref{prop:outer-partition}) of \cref{def:st-outerdec}. Concerning Property (\ref{prop:outer-maximality}), while any of the obtained outerpaths cannot be expanded with the first fan of the next outerpath, it may not hold that $e_i \neq s_it_i$.  Thus, we next show how to (re)assign some fans to a different outerpath in $P_1,\dots,P_m$ so to also satisfy Property (\ref{prop:outer-maximality}) of \cref{def:st-outerdec}.

\begin{figure}[t]
    \centering
    \includegraphics[page=1]{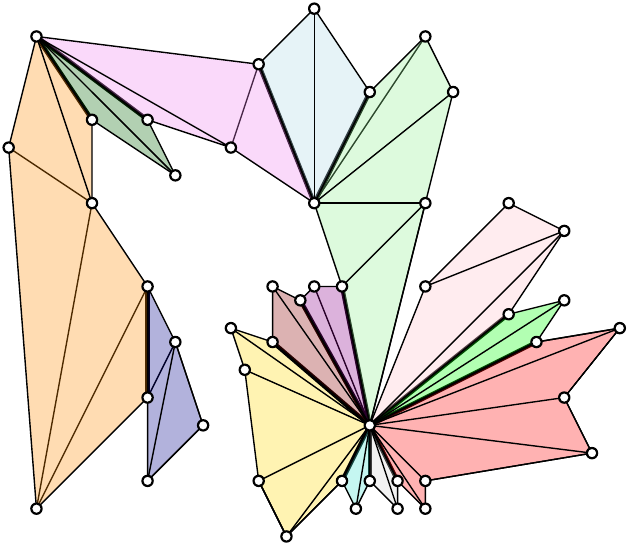}
    \caption{The decomposition in incrementally maximal outerpaths of \cref{fig:st-outer-dec} before the update operation described in the proof of \cref{lem:stouterexsists}.\label{fig:st-outer-dec-before}}
\end{figure}

For $i=1,\dots,m-1$, consider a partition of each $P_i$ into incrementally maximal $st$-fans, obtained by visiting the faces of $P_i$ from the face $f_j$ that contains $e_{i-1}$ to the face $f_{j'}$ that contains $e_i$. Recall that we denoted by $F^i$ the last incrementally maximal $s^i t^i$-fan of $P_i$ and hence $e_i$ belongs to $F^i$; furthermore, let us denote by $Q^i$ the incrementally maximal fan, if any, preceding $F^i$ in the above partition.
Suppose that Property (\ref{prop:outer-maximality}) is not satisfied by $P_1,P_2,\dots,P_m$, and
let $j$ be the minimum index such that $e_j = s^j t^j$.
Consider the two faces of $G$ incident to $e_j$. 
Observe that one of these faces, which we denote by $f'$, belongs to $F^j$ (of $P_j$), while the other face, which we denote by $f''$, belongs to the first incrementally maximal fan of $P_{j+1}$.
First, observe that $P_j \neq F_j$, since $F_j \cup f''$ is a single-source single-sink outerpath (and therefore the edge $s^i t^i$ would not be the shared edge of two distinct outerpaths in $P_1,P_2,\dots,P_m$).
Therefore, $P_j$ also contains $Q^j$. The edge $e$ shared by $Q^j$ and $F^j$ is not $s^j t^j$, as if this were the case, by extending $Q^j$ with the face of $F^j$ incident to $e$ we would obtain a larger single-source single-sink fan, which contradicts the fact that $Q^j$ is incrementally maximal.

We update the decomposition $P_1,P_2,\dots,P_m$ as follows.
We remove from $P_j$ all the vertices and edges of $F^j$, except for $e$ and its end-vertices.
We then set $P_{j+1} = F^j \cup P_{j+1}$, that is, we move the fan $F^j$ from $P_j$ to $P_{j+1}$ (note that $F^j$ may merge with the first incrementally maximal fan of $P_{j+1}$). See, for instance, the blue-orange pair and the pink-cyan pair of outerpaths in \cref{fig:st-outer-dec-before}, and how they are updated in \cref{fig:st-outer-dec}.
Observe that $e' \neq s^j t^j$ is now the unique edge shared by $P_j$ and $P_{j+1}$. Therefore, $s_jt_j \neq e_j$ is now satisfied by the pair $P_j, P_{j+1}$.
Repeating the above update, as long as there exists a pair of single-source single-sink outerpaths in $P_1,P_2,\dots,P_m$ sharing an edge and violating Property (\ref{prop:outer-maximality}), eventually yields the desired $st$-outerpath decomposition.
\end{proof}}

An $st$-outerpath that is not a single $st$-fan is called \emph{proper} in the following. Let $P_1,P_2,\dots,P_m$ be an $st$-outerpath decomposition of an upward outerpath $G$, two proper outerpaths $P_i$ and $P_j$ are \emph{consecutive}, if there is no proper outerpath $P_a$, such that $i < a < j$. We will use the following technical lemmas.

\begin{restatable}[$\star$]{lemma}{lemConsecutiveshare}\label{lem:consecutive-share}
Two consecutive proper outerpaths $P_i$ and $P_j$ share either a single vertex $v$ or the edge $e_i$. In the former case, it holds $j>i+1$, in the latter case it holds $j=i+1$. 
\end{restatable}
\newcommand{\lemConsecutiveshareProof}{\begin{proof}
By \cref{def:st-outerdec}, $P_i$ and $P_j$ share at most one edge. We first prove that $P_i$ and $P_j$ cannot be vertex-disjoint. Namely, if $j=i+1$, then $P_i$ and $P_j$ share an edge by Property (\ref{prop:outer-share-one-edge}) of \cref{def:st-outerdec}. If $j>i+1$, let $P_{i+1},\dots,P_{j-1}$ be the non-proper outerpaths between $P_i$ and $P_j$. Since each of these non-proper outerpaths are single-source single-sink fans, it must be that they either have different sources or different sinks, and hence they all share the same sink or the same source, respectively, and this common vertex is also common to~$P_i$~and~$P_j$.\end{proof}}

\begin{restatable}[$\star$]{lemma}{lemConsecutiveFans}\label{lem:consecutive-fans}
Let $P_i$ and $P_j$ be two consecutive proper outerpaths that share a single vertex $v$. Then each $P_a$, with $i < a < j$, is an $s_at_a$-fan such that $v=s_a$ if $v$ is the tail of $e_i$, and $v=t_a$ otherwise.
\end{restatable}
\newcommand{\lemConsecutiveFansProof}{\begin{proof}
We know that each $P_a$, with $i < a < j$, is an $s_at_a$-fan because $P_i$ and $P_j$ are consecutive. As already observed in the proof of \cref{lem:consecutive-share}, since each $P_a$ is an $s_at_a$-fan, it must be that all of them either have different sources or different sinks, and hence vertex $v$ is either a common sink or a common source for all of them. Consequently, if $v$ is the tail of $e_i$, then $v$ is the source $s_a$ of each $P_a$, otherwise $v$ is the sink $t_a$ of each $P_a$.
\end{proof}}

\begin{restatable}[$\star$]{lemma}{lemConsecutiveV}\label{lem:consecutive-v}
Let $v$ be a vertex shared by a set $\cal P$ of $n_{\cal P}$ outerpaths. Then at most 
two outerpaths of $\cal P$ are such that $v$ is internal for them, and $\cal P$ contains at most four proper outerpaths.
\end{restatable}
\newcommand{\lemConsecutiveVProof}{\begin{proof}
Since all outerpaths of $\cal P$ share vertex $v$, the fact that $\cal P$ contains at most two outerpaths for which $v$ is internal immediately follows by the fact that $G$ has an upward outerplanar embedding. Namely, by bimodality, at most two faces incident to $v$ are such that $v$ is neither the source nor the sink of that face~\cite{DBLP:books/ph/BattistaETT99}. These two faces belong to two different outerpaths $P$ and $P'$ in $\cal P$, and if there existed a third one $P''$ with the property of $v$ being internal, then at least one vertex of $P$ or $P'$ would not belong to the outer face of $G$.

To prove the latter part of the statement, observe that, by \cref{lem:consecutive-share}, it must be that the outerpaths in $\cal P$ form a contiguous subset $P_i,\dots,P_j$, with $j-i=n_{\cal P}$, in the $st$-outerpath decomposition $P_1,\dots,P_m$. Since each outerpath in $\cal P$ shares $v$, in order for $\cal G$ to be a path, it must be that each $P_a$ is a fan, except for possibly $P_i$ and $P_j$, although not necessarily a fan in which the vertex shared by the inner edges is its source vertex. However, we proved above that at most two of these fans are such that $v$ is internal. Hence, the only elements of $\cal P$ that can be proper are $P_i$, $P_j$ and the two for which $v$ is internal.
\end{proof}}

We are now ready to prove the main result of this section.
\begin{restatable}[$\star$]{theorem}{thmMain}\label{thm:main}
Every internally-triangulated upward outerpath $G$ admits a $16$-page \ube.
\end{restatable}
\begin{proof}[Sketch]
Let $P_1,P_2,\dots,P_m$ be an $st$-outerpath decomposition of $G$ (\cref{lem:stouterexsists}). Based on \cref{lem:consecutive-share}, a \emph{bundle} is a maximal set of outerpaths that either share an edge or a single vertex. Let $G_b$ be the graph induced by the first $b$ bundles of $G$ (going from $P_1$ to $P_m$). We prove the statement by induction on the number $l$ of bundles of $G$. In particular, we can prove that $G_l$ admits a $e_{g(l)}$-consecutive $16$-page \ube $\langle \pi^l \sigma^l \rangle$, where $g(l)$ is the greatest index such that $P_{g(l)}$ belongs to $G_l$, and such that each single $s_it_i$-outerpath uses at most $4$ pages. In the inductive case, we distinguish whether the considered bundle contains only two outerpaths that share an edge, or at least three outerpaths that share a vertex. Here, we exploit the crucial properties of  \cref{lem:st-outerpath,lem:consecutive-v,lem:consecutive-fans}, which allow us to limit the interaction between different outerpaths in~terms~of~pages.\end{proof}

\newcommand{\thmMainProof}{\begin{proof}
Let $P_1,P_2,\dots,P_m$ be an $st$-outerpath decomposition of $G$, which exists by \cref{lem:stouterexsists}. Based on \cref{lem:consecutive-share}, a \emph{bundle} is a maximal set of outerpaths that either share an edge or a single vertex. Let $G_b$ be the graph induced by the first $b$ bundles of $G$ (from $P_1$ to $P_m$). We prove by induction on the number $l$ of bundles of $G$ that $G_l$ admits a $e_{g(l)}$-consecutive $16$-page \ube $\langle \pi^l \sigma^l \rangle$, where $g(l)$ is the largest index such that $P_{g(l)}$ belongs to $G_l$, and such that each single $s_it_i$-outerpath uses at most $4$ pages. 

Before giving the inductive proof, we observe that each $s_it_i$-outerpath $P_i$, with $i=1,\dots,m$, admits an $e_i$-consecutive $4$-page \ube $\langle \pi, \sigma \rangle$, which can be obtained by using \cref{lem:st-outerpath}. In particular, let $F_1,\dots,F_k$ be the $st$-fan decomposition given in input to \cref{lem:st-outerpath}, and denote by $s^j$, $t^j$ the source and sink of $F_j$, respectively.  We have three possible cases: (a) $F^i = F_k$, or (b) $F^i=F_1$ (and hence $P_i$ is \good), or (c) $F^i$ belongs to the \app $H_1$ of $P_i$. In all the three cases, the preconditions of \cref{lem:st-outerpath} applies. Namely, in case (a) $s_it_i = s^kt^k$ and $e \neq s_it_i$ because $P_i$ is incrementally maximal; in case (b) $s_it_i = s^1t^1$ and again $e \neq s_it_i$ because $P_i$ is incrementally maximal; in case (c) if must be that $s_it_i$ belongs to $H_1$ (otherwise $H_1$ would not contain $F^i$), and hence clearly $s_it_i \neq s^1t^1$. 

In the base case $l=0$, then $m=1$, and an $e_1$-consecutive $4$-page \ube $\langle \pi^0 \sigma^0 \rangle$ of $P_1$ can be obtained by using \cref{lem:st-outerpath} as observed before.

If $l>0$, let $B=\{P_{g'(l)},\dots,P_{g(l)}\}$ be the $l$-th bundle of $G_l$. Let $G_{l-1}$ be the graph obtained by removing from $G_l$ the outerpaths in $B$, except for $P_{g'(l)}$ (which belongs to the bundle before $B$ as well, if any). Note that $G_{l-1}$ contains $l-1$ bundles. By induction, let $\langle \pi_{l-1}, \sigma_{l-1} \rangle$ be an $e_{g(l-1)}$-consecutive $16$-page \ube of $G_{l-1}$. We distinguish whether $B$ contains only two outerpaths that share an edge, or it contains at least three outerpaths that share a vertex.

\begin{figure}[t]
    \centering
    \includegraphics[page=3]{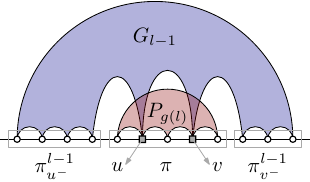}
    \caption{Illustration for \textsf{CASE 1} in the proof of \cref{thm:main}.\label{fig:case-1}}
\end{figure}

\smallskip\noindent\textsf{CASE 1.} $B$ contains only two outerpaths $P_{g(l)-1}$ and $P_{g(l)}$ that share an edge $e_{g(l)-1}=uv$. Let $\langle \pi,\sigma \rangle$ be a $e_{g(l)}$-consecutive $4$-page \ube of $P_{g(l)}$, obtained by applying \cref{lem:st-outerpath}, as already observed. 
Referring to \cref{fig:case-1}, we set $\pi^l=\pi^{l-1}_{u^-} \circ \pi \circ \pi^{l-1}_{v^+}$.
Clearly, $\pi^l$ extends both $\pi^{l-1}$ and $\pi$. 
One easily verifies that if $e$ and $e'$ are such that $e \in G_{l-1}$ and $e' \in P_{g(l)}$, then $e$ is incident to either $u$ or $v$. Consequently, $e$ belongs to $P_{g(l)-1}$. Namely, if $e$ belonged to another outerpath $P \neq P_{g(l)-1}$, then $P,P_{g(l)-1}, P_{g(l)}$ would all share a vertex and hence would belong to the same bundle $B$ by the maximality of $B$. Hence we can set $\sigma^l(e)=\sigma^{l-1}(e)$ for each edge $e \in G_{l-1}$, while, for the edges of $P_{g(l)}$, we can use any set of $4$ pages in $\sigma^{l-1}$ not used by the edges of $P_{g(l)-1}$.  

\begin{figure}[t]
    \centering
    \includegraphics[page=4,width=\textwidth]{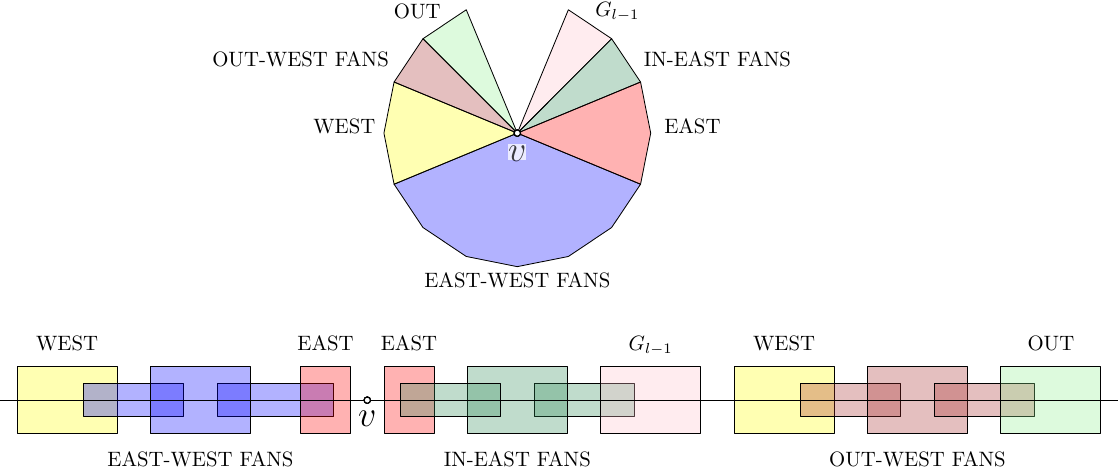}
    \caption{Illustration for \textsf{CASE 2} in the proof of \cref{thm:main}. The top figure shows the cyclic order of the outerpaths in $B$ around $v$. The bottom figure shows the relative ordering of the vertices of the different outerpaths. We remark that the colors used here are unrelated with those of \cref{fig:st-outer-dec}.\label{fig:case-2}}
\end{figure}

\smallskip\noindent\textsf{CASE 2.} $B$ contains at least three outerpaths that share a vertex. For ease of notation, let us denote $P_i=P_{g'(l)}$ and $P_{j}=P_{g(l)}$, and denote by $v$ the shared vertex. By \cref{lem:consecutive-v}, we know that in $P_i,\dots,P_j$ there are at most four proper outerpaths, and, by \cref{lem:consecutive-fans}, that between any two consecutive proper outerpaths there are non-proper outerpaths for which $v$ is either a common source or a common sink. We assume that all these four proper outerpaths exist, as the proof is simpler otherwise. For ease of notation, we denote by IN and OUT the first and the last proper outerpaths (in the order from $P_1$ to $P_m$). Note that $G_{l-1}$ contains IN. Also, we denote by EAST and WEST the proper outerpaths such that EAST comes before WEST. Moreover, we denote by IN-EAST fans, EAST-WEST fans, and WEST-OUT fans, the non-proper outerpaths between the IN and EAST, EAST and WEST, WEST and OUT, respectively. For each of them, we compute a \ube by \cref{lem:st-outerpath} as in \textsf{CASE 1}, and we merge the obtained \ube{s} one after the other as in  \textsf{CASE 1}. A schematic illustration of the obtained ordering is illustrated in \cref{fig:case-2}. One can verify that the edges of two fans do not cross unless they belong to the same set (IN-EAST, EAST-WEST, OUT-WEST) and share an edge. Thus, for all non-proper outerpaths we will use two alternating sets of $2$ pages each (rather than $4$, as these are simple fans), different from the $4$ pages used by IN. The edges of $G_{l-1}$ can cross only with the edges of the neighboring IN-EAST fan or with the edges of EAST. Since the IN-EAST fans use different pages, we only need to argue about the edges of EAST. Such edges only cross edges of $G_{l-1}$ that are incident to $v$, hence, as in \textsc{CASE 1}, such edges belong to IN by the maximality of $B$. Thus, we can use a set of $4$ pages for the edges of EAST different from the one used by IN. The edges of WEST only cross one neighboring EAST-WEST fan and one neighboring OUT-WEST fan, as well as the edges of EAST.  Hence, we can use the same $4$ pages used by IN. Similarly, the OUT edges only cross the edges of one neighboring OUT-WEST fan, the edges of EAST, and the edges of WEST. Thus we can use the last set of $4$ pages different from those used by IN (which are the same used by WEST), by EAST, and by the fans. Overall, the computed \ube still uses at most $16$ pages.\end{proof}}

\subsection{Upward outerplanar graphs}\label{sse:outerplanar}

We now deal with upward outerplanar graphs that may be non-triangulated and may have multiple sources and sinks, but whose blocks are $st$-DAGs. We begin with the following lemma, which generalizes \cref{lem:st-outerpath} in terms of UBT.

\begin{restatable}[$\star$]{lemma}{lemStOuterplanar}\label{lem:stouterplanar}Every biconnected $st$-outerplanar graph $G$ admits a $4$-page \ube.
\end{restatable}
\begin{proof}[Sketch]
By exploiting a technique in~\cite{DBLP:journals/tcs/BattistaT88}, we can assume that $G$ is internally triangulated. Let $\langle f_1,\dots,f_h \rangle$ be a path in $\overline{G}$ whose primal graph $P \subset G$ is a \good $st$-outerpath. Each outer edge $uv$ of $P$ is shared by $P$ and by a one-sided $uv$-outerpath. Then a $4$-page \ube of $G$ exists by \cref{lem:one-insertion}.
\end{proof}
\newcommand{\lemStOuterplanarProof}{\begin{proof}
\begin{figure}[t]
    \centering
\includegraphics[page=1,width =.45\textwidth]{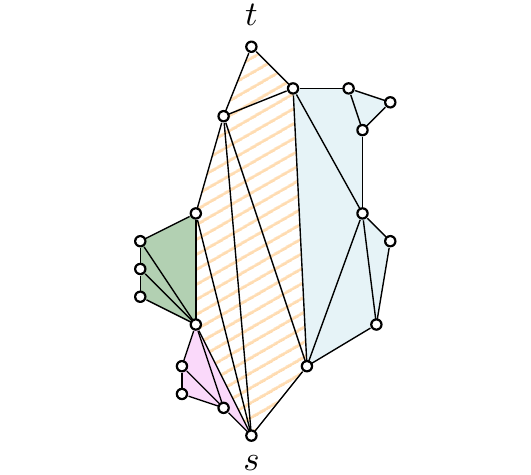}
    \caption{Illustration for the proof of \cref{lem:stouterplanar}.  An $st$-outerplanar graph $G$; an $st$-outerpath $P$ of $G$ is highlighted with a striped orange background, while the three corresponding subgraphs are light purple, green, and blue. \label{fig:stouterplanar-a}}
\end{figure}
Di Battista and Tamassia~\cite{DBLP:journals/tcs/BattistaT88} proved that every upward planar graph $G$ can be augmented, by only introducing edges, to an upward planar triangulation. In fact, their technique can also be applied to augment an upward outerplanar graph to an internally-triangulated upward outerplanar graph. Thus, we can assume that $G$ is internally triangulated. 

Refer to \cref{fig:stouterplanar-a}. Let $f_1$ be any inner face of $G$ that contains vertex $s$, and let $f_h$ be any inner face of $G$ that contains vertex $t$. Let $\Pi = \langle f_1,\dots,f_h \rangle$ be a path in $\overline{G}$. By construction, the primal graph $P$ of $\Pi$ is a \good $st$-outerpath. Each edge $uv$ on the outer face of $P$ either belongs to the outer face of $G$ or $uv$ is an inner edge of $G$ and therefore $\{u,v\}$ is a separation pair of $G$. We denote by $H_{uv}$ the maximal connected component of $G$ that contains $uv$ and no further vertex of $P$ (such component is unique by outerplanarity). Since $G$ has a single source and a single sink, $H_{uv}$ is an outerplanar graph with a single source $u$ and a single sink $v$. Moreover, since edge $uv$ is on the outer face of $H_{uv}$, it follows that $H_{uv}$ is a one-sided $uv$-outerplanar graph. Then a  $4$-page \ube of $G$ exists by  \cref{lem:one-insertion}.
\end{proof}}

We are now ready to show the main result of this subsection.

\begin{figure}[t]
    \centering
\includegraphics[page=1]{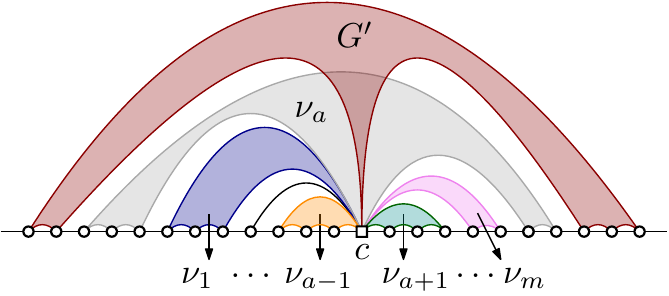}
    \caption{Illustration for the proof of \cref{thm:general}.  \label{fig:blocks}}
\end{figure}

\begin{restatable}[$\star$]{theorem}{thmGeneral}\label{thm:general}
Every upward outerplanar graph $G$ whose biconnected components are $st$-outerplanar graphs admits an $8$-page \ube.
\end{restatable}
\begin{proof}[Sketch]
 We prove a stronger statement. Let $T$ be a BC-tree of $G$ rooted at an arbitrary block $\rho$, then $G$ admits an $8$-page \ube $\langle \pi, \sigma \rangle$ that has the \emph{page-separation} property: For any block $\beta$ of $T$, the edges of $\beta$ are assigned to at most 4 different pages.  We proceed by induction on the number $h$ of cut-vertices in $G$. If $h=0$, then $G$ consists of a single block and the statement follows by \cref{lem:stouterplanar}. Otherwise, let $c$ be a cut-vertex whose children are all leaves. Let $\nu_1,\dots,\nu_m$ be the $m>1$ blocks representing the children of $c$ in $T$, and let $\mu$ be the parent block of $c$. Also, let $G'$ be the maximal subgraph of $G$ that contains $\mu$ but does not contain any vertex of $\nu_1,\dots,\nu_m$ except $c$, that is, $G=G'\cup \nu_1 \cup \dots \cup \nu_m$. By induction, $G'$ admits an $8$-page \ube $\langle \pi', \sigma' \rangle$ for which the page-separation property holds, as it contains at most $h-1$ cut-vertices. On the other hand, each $\nu_i$ admits a $4$-page \ube $\langle \pi^i, \sigma^i \rangle$ by \cref{lem:stouterplanar}. By \cref{lem:bimodality}, at most two blocks in $\{\mu\} \cup \{\nu_1,\dots,\nu_m\}$ are such that $c$ is internal.

Let us assume that there are exactly two such blocks and one of these two blocks is $\mu$, as otherwise the proof is just simpler. Also, let $\nu_a$, for some $1 \le a \le m$ be the other block for which $c$ is internal. Up to a renaming, we can assume that $\nu_1,\dots,\nu_{a-1}$ are $st$-outerplanar graphs with sink $c$, while $\nu_{a+1},\dots,\nu_m$ are $st$-outerplanar graphs with source $c$. Refer to \cref{fig:blocks}. Crucially, we set: 

$$
\pi=\pi'_{c^-} \cup \pi^a_{c^-} \cup \pi^1_{c^-} \cup \dots \cup \pi^{a-1}_{c^-} \cup \{c\} \cup \pi^{a+1}_{c^+} \cup \dots \cup \pi^{m}_{c^+} \cup \pi^a_{c^+} \cup \pi'_{c^+}.    
$$

The page assignment is based on the fact that $e$ and $e'$, such that $e \in \nu_i$ and $e' \in G'$, cross each other only if $i=a$ and in such a case $e'$ is incident to $c$.
\end{proof}
\newcommand{\thmGeneralProof}{
\begin{proof}
 We prove a stronger statement. Let $T$ be a BC-tree of $G$ rooted at an arbitrary block $\rho$, then $G$ admits an $8$-page \ube $\langle \pi, \sigma \rangle$ that has the \emph{page-separation} property: For any block $\beta$ of $T$, the edges of $\beta$ are assigned to at most 4 different pages.  We proceed by induction on the number $h$ of cut-vertices in $G$.

If $h=0$, then $G$ consists of a single block which is an $st$-outerplanar graph and the statement follows by \cref{lem:stouterplanar}. 

Otherwise, let $c$ be a cut-vertex whose children are all leaves. Let $\nu_1,\dots,\nu_m$ be the $m>1$ blocks representing the children of $c$ in $T$, and let $\mu$ be the parent block of $c$. Also, let $G'$ be the maximal subgraph of $G$ that contains $\mu$ but does not contain any vertex of $\nu_1,\dots,\nu_m$ except $c$, that is, $G=G'\cup \nu_1 \cup \dots \cup \nu_m$. By induction, $G'$ admits an $8$-page \ube $\langle \pi', \sigma' \rangle$ for which the page-separation property holds, as it contains at most $h-1$ cut-vertices. On the other hand, each $\nu_i$ admits a $4$-page \ube $\langle \pi^i, \sigma^i \rangle$ by \cref{lem:stouterplanar}. By \cref{lem:bimodality}, at most two blocks in $\{\mu\} \cup \{\nu_1,\dots,\nu_m\}$ are such that $c$ is internal. 

Let us assume that there are exactly two such blocks and one these two blocks is $\mu$, as otherwise the proof is just simpler. Also, let $\nu_a$, for some $1 \le a \le m$ be the other block for which $c$ is internal. Up to a renaming, we can assume that $\nu_1,\dots,\nu_{a-1}$ are $st$-outerplanar graphs with sink $c$, while $\nu_{a+1},\dots,\nu_m$ are $st$-outerplanar graphs with source $c$. Refer to \cref{fig:blocks}. We set: $$\pi=\pi'_{c^-} \cup \pi^a_{c^-} \cup \pi^1_{c^-} \cup \dots \cup \pi^{a-1}_{c^-} \cup \{c\} \cup \pi^{a+1}_{c^+} \cup \dots \cup \pi^{m}_{c^+} \cup \pi^a_{c^+} \cup \pi'_{c^+}.$$ 
Observe that $\pi$ extends each $\pi^i$, as well as $\pi'$. 
Then one easily verifies that no two edges $e$ and $e'$ such that $e \in \nu_i$ and $e' \in \nu_j$, for any $i \neq j$, cross each other in $\pi$. Similarly,  two edges $e$ and $e'$ such that $e \in \nu_i$ and $e' \in G'$, cross each other only if $i=a$ and in such a case $e'$ is incident to $c$. Concerning the page assignment, we let $\sigma(e) = \sigma'(e)$ for each edge $e$ of $G'$. Also, we let $\sigma(e)=\sigma^i(e)$ for each edge $e$ of $\nu_i$, except the edges of $\nu^a$, assuming that the set of pages of $\sigma^i$ is the same (up to a renaming) as the one used in $\sigma'$ for the edges of $\mu$  (whose size is $4$ by the page-separation property).  For the edges of $\nu_a$, we set $\sigma(e)=\sigma^a(e)$, assuming that the set of pages of $\sigma^a$ is the same (up to a renaming) as the one in $\sigma'$ {\em not} used for the edges of $\mu$ (whose size is again four). As we already observed, if two edges $e$ and $e'$ such that $e \in \nu_a$ and $e' \in G'$ cross each other, then $e'$ is incident to $c$ and hence belongs to $\mu$. Then our assignment guarantees $\sigma(e) \neq \sigma(e')$. Also, the computed $8$-page \ube respects the page-separation property.
\end{proof}}

\subsection{Upward cactus graphs}\label{sse:cactus}

\begin{figure}[tb!]
    \centering
    \subfigure[]{
    \centering
    \includegraphics[page=1]{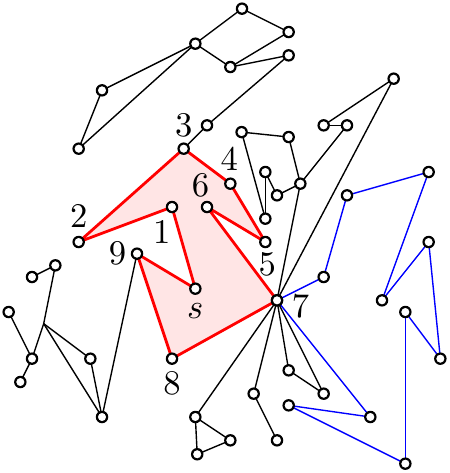}\label{fig:cactus}
    }
    \hfil
    \subfigure[]{
    \includegraphics[page=2]{figs/cactus.pdf}\label{fig:cycle}
    }
    \caption{(a) A cactus $G$ and (b) a 2-page \ube of the red non-trivial block of $G$.}
\end{figure}

The first lemma allows us to consider cactus graphs with no trivial blocks. 

\begin{restatable}[$\star$]{lemma}{lemNoTrivial}\label{lem:no-trivial}A cactus $G'$ can always be augmented to a cactus $G$ with no trivial blocks and such that the embedding of $G'$ is maintained.
\end{restatable}

\newcommand{\lemNoTrivialProof}{
\begin{proof}
If the blocks of $G'$ are cycles, then there is nothing to be done. Otherwise, $G'$ contains a trivial block $\beta$ corresponding to the directed edge $uv$. 
We show how to extend $G'$ to a cactus $G''$, whose number of trivial blocks is one less than the number of trivial blocks of $G'$. For simplicity, we construct $G''$, by showing how to extend an upward outerplanar drawing $\Gamma'$ of $G'$ to an upward outerplanar drawing $\Gamma''$ of $G''$.
First, we initialize $\Gamma''= \Gamma'$. 
Second, we place in $\Gamma''$ a new vertex $w$ above $u$ and below $v$ and arbitrarily close to the drawing of the edge $uv$. Finally, we draw two new edges $uw$ and $wv$ as $y$-monotone curves so that $uv$, $uw$ and $wv$ form a face of the drawing. Clearly, $\Gamma''$ is an upward outerplanar drawing of $G''$, $G''$ is a cactus, and $G''$ contain one less trivial block of $G'$. By repeating the process until no trivial blocks exist, we eventually obtain $G$.\end{proof}
}

\noindent It is well known that any DAG whose underlying graph is a cycle admits a $2$-page \ube~\cite[Lemma~2.2]{HPT99a}. We can show a slightly stronger result, which will prove useful afterward; see \cref{fig:cycle}.

\begin{restatable}[$\star$]{lemma}{lemCycle}\label{lem:cycle}Let $G$ be a DAG whose underlying graph is a cycle and let $s$ be a source (resp. let $t$ be a sink) of $G$. Then, $G$ admits a $2$-page \ube $\langle \pi, \sigma \rangle$ where $s$ is the first vertex (resp.\ $t$ is the last vertex) in $\pi$.\end{restatable}

\newcommand{\lemCycleProof}{
\begin{proof}
Let $sw$ be any of the two edges incident to~$s$. Let $G'$ be the path $G - sw$. Since $G'$ is a tree, it admits a $1$-page \ube $\langle \pi', \sigma' \rangle$ in which $s$ is the first vertex of $\pi'$; see \cref{fig:cycle}. We obtain a $2$-page \ube $\langle \pi, \sigma \rangle$ of $G$, by setting $\pi = \pi'$, $\sigma(e)=\sigma'(e)=1$ for each edge $e \neq sw$, and $\sigma(sw)=2$.
A symmetric argument can be used to prove the statement with respect to the vertex $t$.
\end{proof}}

\noindent
Using the proof strategy of \cref{thm:general}, we can exploit \cref{lem:cycle} to show: %

\begin{restatable}[$\star$]{theorem}{thmMainCactus}\label{thm:MainCactus}
Every upward outerplanar cactus $G$ admits a $6$-page \ube.
\end{restatable}
\begin{proof}[Sketch]
By \cref{lem:no-trivial}, we can assume that all the blocks of $G$ are non-trivial, i.e., correspond to cycles.
Also, let
$T$ be the BC-tree of $G$ rooted at any block. The theorem can be proved by induction on the number of blocks in $T$. In fact, we prove the following slightly stronger statement: 
$G$ admits a $6$-page \ube in which the edges of each block lie on at most two  pages. The proof crucially relies on \cref{lem:bimodality} and follows the lines of \cref{thm:general}.\end{proof}
\newcommand{\thmMainCactusProof}{\begin{proof}
By \cref{lem:no-trivial}, we can assume that all the blocks of $G$ are non-trivial, i.e., correspond to cycles.
Also, let
$T$ be the BC-tree of $G$ rooted at any block $\rho$. We prove the theorem by induction on the number $\beta(T)$ of blocks in $T$. In fact, we prove the following slightly stronger statement: 
$G$ admits a $6$-page \ube in which the edges of each block lie on at most two pages.

If $\beta(T)=1$, then $G$ consists of a single directed cycle, and the statement follows from \cite[Lemma~2.2]{HPT99a}.

If $\beta(T)>1$, then $T$ contains at least one cut-vertex $c$ whose children blocks $\nu_1,\dots,\nu_k$ are all leaves of $T$. Let $\mu$ be the parent block of $c$.
Consider the subgraph of $G'$ induced by all the edges of the blocks of $G$ different from $\nu_1,\dots,\nu_k$.  Note that, 
by induction, $G'$ admits a $6$-page \ube $\langle \pi', \sigma' \rangle$ in which the edges of $\mu$ lie on two distinct pages, say pages $1$ and $2$.

We show how to extend $\langle \pi',\sigma' \rangle$ to a $6$-page book embedding $\langle \pi,\sigma \rangle$ of $G$ in which the edges of each block $\nu_i$, for $i=1,\dots,k$, lie on two distinct pages, namely, either pages $1$ and $2$, or pages $3$ and $4$, or pages $5$ and $6$.

Recall that, by \cref{lem:bimodality}, there exist at most two blocks among $\nu_1,\dots,\nu_k$ for which $c$ is internal. %
After a possible renaming of the indexes, we may assume that
\begin{inparaenum}[(i)]
\item  $\nu_1$ and $\nu_2$ are the blocks, if any, for which $c$ is internal; let $0 \leq r \leq 2$ be the number of such children;
\item
$\nu_{r+1},\dots,\nu_h$ with $h \leq k$ are the blocks, if any, for which $c$ is a sink; and
\item $\nu_{h+1},\dots,\nu_k$ are the blocks, if any, for which $c$ is a source.
\end{inparaenum}

\begin{figure}[htb!]
    \centering
    \includegraphics[page=3]{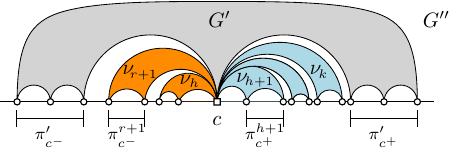}
    \caption{Construction of $\langle \pi'', \sigma'' \rangle$ from $\langle \pi', \sigma' \rangle$; edges incident to $c$ lie~on~two~pages.}\label{fig:monotonic}
\end{figure}
Consider the graph $G''$ induced by the edges of $G'$ and of the blocks $\nu_i$ with ${r+1} \leq i\leq k$.
First, we show how to extend $\langle \pi',\sigma' \rangle$ to a $6$-page \ube $\langle \pi'',\sigma'' \rangle$ of $G''$ in 
which the edges of the blocks $\nu_i$ with $r+1 \leq i\leq k$ lie on pages $1$ and $2$.
Refer to \cref{fig:monotonic}.
Let $\pi'= \pi'_{c^-} \circ \langle c \rangle \circ  \pi'_{c^+}$, where either $\pi'_{c^-}$  or $\pi'_{c^+}$ may be empty (recall that each block of $G$ is a cycle, and thus it contains at least three vertices).
For each $i=r+1,\dots,h$ (resp.\ $i=h+1,\dots,k$), construct a $2$-page \ube $\langle \pi^i, \sigma^i \rangle$ of $\nu_i$ on pages $1$ and $2$ in which $c$ is the last vertex (resp.\ the first vertex) of $\pi$, by \cref{lem:cycle}
We obtain $\pi''$ as follows. 
We set $\pi'' = \pi'_{c^-} \circ \pi^{r+1}_{c^-} \circ \dots \circ \pi^h_{c^-} \circ \langle c \rangle \circ \pi^{h+1}_{c^+} \circ \cdots \circ \pi^{k}_{c^+} \circ \pi'_{c^+}$. 
Further, we obtain $\sigma''$ by setting $\sigma''(e)=\sigma'(e)$, for any edge of $G'$, and by setting $\sigma''(e)=\sigma_i(e)$, for any edge of $\nu_i$ with ${r+1} \leq i \leq k$. 
By construction, no two edges belonging to distinct blocks among $\nu_{r+1},\dots,\nu_k$ intersect. Therefore, since the edges these blocks lie on the same two pages the edges of the block $\mu$ lie on, namely, pages $1$ and $2$, we have that
$\langle \pi'', \sigma'' \rangle$ is a $6$-page \ube with~the~desired~properties.

If $r=0$, by setting $\pi = \pi''$ and $\sigma = \sigma''$, we obtain the desired $6$-page book embedding of $G$. Next, we show how to obtain $\langle \pi,\sigma \rangle$ from $\langle \pi'',\sigma'' \rangle$, when $r>0$.
We distinguish two cases, based on whether ({\sf CASE A}) $c$ is a source or a sink of $\mu$, or  ({\sf CASE B}) $c$ is internal for $\mu$. Let $\pi'' = \pi''_{c^-} \circ \langle c \rangle \circ \pi''_{c^+}$. Also, for $i \leq r$, we construct a $2$-page \ube $\langle \pi_i, \sigma_i \rangle$ of $\nu_i$ on pages $2i+1$ and $2i+2$, by means of \cref{lem:cycle}, and let $\pi^i = \pi^i_{c^-} \circ \langle c \rangle \circ \pi^i_{c^+}$.

\begin{itemize}
    \item[{\sf CASE A}] Suppose that $c$ is a source of $\mu$; the case in which $c$ is a sink of $\mu$ being analogous. Then, $1 \leq r \leq 2$. We assume $r=2$, the case $r=1$ being simpler. Note that, 
    the vertex $c$ may be 
    internal for $\mu \cup \nu_3 \cup \dots \cup \nu_k$, as illustrated by the gray shaded region depicting this subgraph in \cref{fig:r2}.
We set $\pi=\pi''_{c^-} \circ \pi^2_{c^-}  \circ \pi^1_{c^-} \circ \{c\}  \circ \pi^1_{c^+} \circ \pi^2_{c^+} \circ \pi''_{c^+}$. Further, we obtain $\sigma$ by setting $\sigma(e)=\sigma''(e)$, for any edge of $G''$, and by setting $\sigma''(e)=\sigma_i(e)$, for any edge of $\nu_i$ with $i \in \{1,2\}$.
By construction, the edges of $\nu_1$ (resp.\ of $\nu_2$) may only intersect the edges of $\nu_{2}$ (resp.\ $\nu_1$) and the edges of $\mu \cup \nu_{3} \cup \dots \cup \nu_k$.
Therefore, since the edges of $\mu \cup \nu_{3} \cup \dots \cup \nu_k$ lie on pages $1$ and $2$, the edges of $\nu_1$ lie on pages $3$ and $4$, and the edges of $\nu_2$ lie on pages $5$ and $6$, we have that
$\langle \pi, \sigma \rangle$ is a $6$-page \ube with the desired properties.
\begin{figure}[htb!]
    \centering
    \includegraphics[page=4]{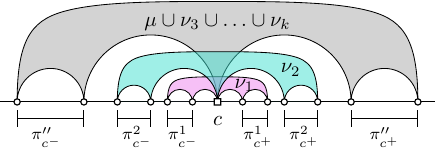}
    \caption{Construction of $\langle \pi, \sigma \rangle$ from $\langle \pi'', \sigma'' \rangle$ focused on $\mu$ and~its~children ($r=2$).}\label{fig:r2}
\end{figure}
\item[{\sf CASE B}] 
Suppose now that $c$ is  internal for $\mu$. Then, $r = 1$. 
First, we construct a $2$-page \ube $\langle \pi_1, \sigma_1 \rangle$ of $\nu_1$ on pages $3$ and $4$, by means of \cref{lem:cycle}, and let $\pi_1 = \pi^1_{c^-} \circ \langle c \rangle \circ \pi^1_{c^+}$.
Then, we set $\pi= \pi''_{c^-} \circ \pi^1_{c^-} \circ \langle c \rangle  \circ \pi^1_{c^+} \circ \pi''_{c^+}$. Further, we obtain $\sigma$ by setting $\sigma(e)=\sigma''(e)$, for any edge of $G''$, and by setting $\sigma(e)=\sigma_1(e)$, for any edge of $\nu_1$.
By construction, the edges of $\nu_1$ may only intersect the edges of $\mu \cup \nu_{2} \cup \dots \cup \nu_k$.
Thereofore, since the edges of $\mu \cup \nu_{2} \cup \dots \cup \nu_k$ lie on pages $1$ and $2$ and the edges of $\nu_1$ lie on pages $3$ and $4$ lie on distinct pairs of pages, we have that
$\langle \pi, \sigma \rangle$ is a $6$-page \ube with the desired properties.
\end{itemize}
\noindent This concludes the proof of the theorem.
\end{proof}}

\section{Complexity Results}\label{se:complexity}

Recall that the upward book thickness problem is \NP-hard for any fixed $k \geq 3$~\cite{DBLP:conf/compgeom/BinucciLGDMP19}. This implies that the problem is para-\NP-hard, and thus it belongs neither to the \textsf{FPT} class nor to the \XP\ class, when parameterized by its natural parameter (unless \P=\NP). In this section, we investigate the parameterized complexity of the problem with respect to the domination number and the vertex cover number, showing a lower and an upper bound, respectively.

\subsection{Hardness Result for Graphs of Bounded Domination Number}\label{sse:hardness}

A \emph{domination set} for a graph $G=(V,E)$ is a subset  $D \subseteq V$ such that every vertex in $V \setminus D$ has {\em at least one} neighbor in $D$. The \emph{domination number} $\gamma(G)$ of $G$ is the number of vertices in a smallest dominating set for $G$. Given a DAG $G$ such that $\ubt{G} \leq k$, one may consider the trivial reduction obtained by considering the DAG $G'$ obtained from $G$ by introducing a new super-source (connected to all the vertices), which has domination number $1$ and for which it clearly holds that $\ubt{G'} \leq k+1$. However, the other direction of this reduction is not obvious, and indeed for this to work we show a more elaborated~construction.

\begin{restatable}[$\star$]{theorem}{lemReduction}\label{lem:reduction}
Let $G$ be an $n$-vertex DAG and let $k$ be a positive integer. It is possible to construct in $O(n)$ time an $st$-DAG $G'$  with $\gamma(G') \in O(k)$ such that $\ubt{G} \leq k$ if and only if  $\ubt{G'}=k+2$. 
\end{restatable}

\begin{proof}[Sketch]
\begin{figure}[tb!]
    \includegraphics[page=1,width=\textwidth]{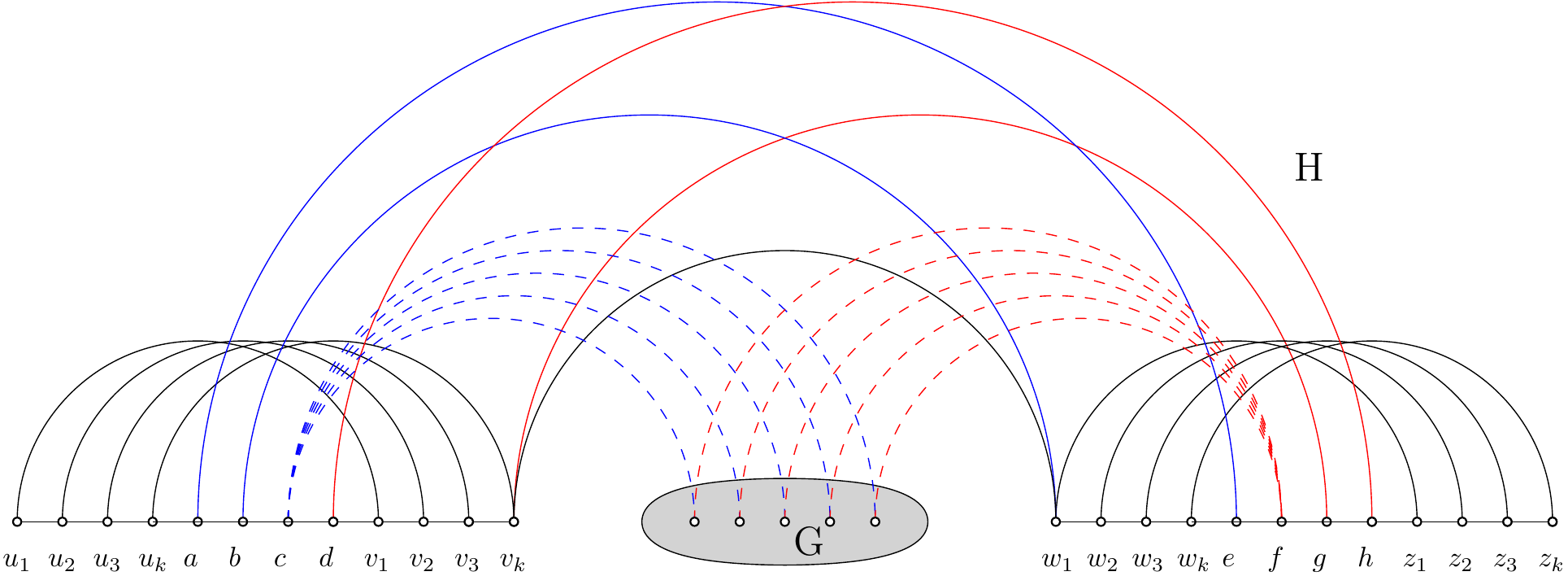}
    \caption{The graph $G'$ in the reduction of \cref{lem:reduction}. The edges of the auxiliary graph $H$ are solid. The black edges lie in $k$  pages.}
    \label{fig:dominating}
\end{figure}
The proof is based on the construction in \cref{fig:dominating}. We obtain graph $G'$ by suitably combining $G$ with an auxiliary graph $H$ whose vertices have the same order in any \ube of $H$, and $\ubt{H}=k+2$. The key property of $G'$ is that the vertices of $G$ are incident to vertices $c$ and $f$ of $H$, and that edges incident to each of these vertices \mbox{must lie in the same page in any \ube of $G'$.}\end{proof}

\newcommand{\lemReductionProof}{\begin{proof}
First, we construct an auxiliary graph $H$ as follows; refer to \cref{fig:dominating}. 
We initialize $H$ as the union of two directed paths
$p_1 = (u_1, u_2, \dots, u_k, a, b, c, d, v_1, v_2,$ \dots,$ v_k)$ and 
$p_2 = (w_1,w_2,\dots, w_k,e,f,g,h,z_1,z_2,\dots,z_k)$.
Further, we add to $H$ the sets of edges $E_1 := \{u_iv_i: 1 \leq i \leq k\}$ and $E_2:= \{ w_iz_i: 1 \leq i \leq k\}$. 
Then, we add the edges 
$ae$, $bw_1$, $dh$, $v_k w_1$, and $v_k g$.
Clearly, $H$ is an $st$-DAG with source $u_1$ and sink $z_k$. Moreover, since all the vertices of $H$ lie in a directed path from $u_1$ to $z_k$ their order $\pi_H$ is the same in any \ube of $H$. Note that, in $\pi_H$, the edges in $E_1 \cup \{ae,dh\}$ pairwise cross, and the same holds for the edges in $E_2 \cup \{ae,dh\}$. This implies that $\ubt{H} \geq k+2$. 
In fact, it is immediate to see that $\ubt{H}=k+2$.
We have the following properties for $H$. 
\begin{property}\label{prop:same-H}
In any  $(k+2)$-page \ube $\langle \pi_H, \sigma_H \rangle$ of $H$ it holds that
 $\sigma_H(ae)=\sigma_H(bw_1)$ and $\sigma_H(dh)=\sigma_H(v_k g)$.
\end{property}

We obtain $G'$ as follows. First, we initialize $G'$ to the union of $G$ and $H$. Then, we add to $G'$ the edges $c v$ and $vf$, for any vertex $v$ of $G$. Clearly, $G'$ is an $st$-DAG with source $u_1$ and sink $z_k$, and can be computed in $O(n)$ time.

First, we show that $\gamma(G') \in O(k)$. Note that $p_1$ and $p_2$ each contain $2k+4$ vertices. Therefore, it is possible to dominate each such path by selecting $k+2$ of its vertices. Further, we can dominate each vertex of $G$ by selecting either the vertex $c$ or the vertex $f$.

Next, we show that $G$ admits a $k$-page \ube $\langle \pi, \sigma\rangle$ if and only if $G'$ admits a $(k+2)$-page \ube $\langle \pi', \sigma' \rangle$.
Suppose first that $G$ admits a $k$-page \ube $\langle \pi, \sigma\rangle$. We obtain a $(k+2)$-page \ube $\langle \pi', \sigma' \rangle$ as follows.
Let $\pi^1_H$ and $\pi^2_H$ be the order of the vertices of $p_1$ and of $p_2$ in $\pi_H$, respectively. We set $\pi' = \pi^1_H \circ \pi \circ \pi^2_H$. Further, 
we assign the edges of $G$ to the same (up to $k$) pages in $\sigma'$ as in $\sigma$; then,
we set 
(i) $\sigma'(u_i v_i) = \sigma'(w_i z_i)=i$, for $i=1,\dots,k$, (ii) $\sigma'(ae) = \sigma'(bw_1) = k+1$ and $\sigma'(dh) = \sigma'(v_kg) = k+2$, and (iii) $\sigma'(cv)=k+1$ and $\sigma'(vf)=k+2$ for any vertex $v$ of $G$. This concludes the construction of $\langle \pi', \sigma' \rangle$. To see that $\langle \pi', \sigma' \rangle$ is indeed a $(k+2)$-page \ube of $G'$ observe, in particular, that $\sigma'$ fulfils  \cref{prop:same-H} and that the only edges of $E(G') \setminus E(G)$ that cross the edges of $G$ are the edges $cv$ and $vf$, where $v$ is a vertex of $G$, and that such edges are assigned to the pages $k+1$ and $k+2$, respectively, which are not used by the edges of $G$.

Suppose now that $G'$ admits a $(k+2)$-page \ube $\langle \pi', \sigma'\rangle$. We obtain a $k$-page \ube $\langle \pi, \sigma \rangle$ as follows. First, we show the following.

\begin{property} \label{claim:verticesofG}
In $\pi'$, all the vertices of $G$ lie after $v_k$ and before $w_1$.
\end{property}

\begin{proof}
Observe that, all the vertices of $G$ must trivially lie between $c$ and $f$, since edges $cv$ and $vf$ exist in $G'$ for any vertex $v$ of $G$. Suppose, for a contradiction, that a vertex $v$ of $G$ lies after $w_1$. Then, the edge $cv$ would cross (at least) the following $k+2$ edges: all the edges of $E_1$, the edge $dh$, and the edge $bw_1$. However, the mentioned edges pairwise cross, which yields a contradiction. The fact that no vertex of $G$ lies before $v_k$ can be proved analogously.
\end{proof}

By \cref{claim:verticesofG,prop:same-H}, we have that for each vertex $v$ of $G$, it holds that $\sigma'(cv)=\sigma'(ae)$ (which, w.l.o.g., we let be $k+1$) and $\sigma'(vf)=\sigma'(dh)$ (which, w.l.o.g., we let be $k+2$). Therefore, $\sigma'$ may place in pages $k+1$ and $k+2$ only edges connecting two consecutive vertices of $G$. Thus, by redefining $\sigma'(e)=1$ for any such an edge, we obtain a $(k+2)$-page \ube of $G'$ in which the edges of $G$ are assigned to $k$ pages. Removing the vertices of $H$ and their incident edges yields the desired $k$-page \ube $\langle \pi, \sigma \rangle$ of $G$.
\end{proof}}

Since testing for the existence of a $k$-page \ube is \NP-hard \mbox{when $k\geq 3$~\cite{DBLP:conf/compgeom/BinucciLGDMP19},} \cref{lem:reduction} implies that the problem remains \NP-hard even for inputs whose domination number is linearly bounded by $k$. We formalize this in the following.

\begin{theorem}\label{th:domination}
For any fixed $k \geq 5$, deciding whether an $st$-DAG $G$ is such that $\ubt{G}\le k$ is \NP-hard even if $G$ has domination number at most $O(k)$. 
\end{theorem}

\cref{th:domination} immediately implies that the upward book thickness problem parameterized by the domination number is para-\NP-hard. On the positive side, we next show that the problem parameterized by the vertex cover number admits a kernel and hence lies in the \textsf{FPT} class.

\subsection{FPT Algorithm Parameterized by the Vertex Cover Number}\label{sse:fpt}

We prove that the upward book thickness problem parameterized by the vertex cover number admits a (super-polynomial) kernel. We build on ideas in~\cite{DBLP:journals/jgaa/BhoreGMN20}. A \emph{vertex cover} of a graph $G=(V,E)$ is a subset  $C \subseteq V$ such that each edge in $E$ has at least one incident vertex in $C$ (a vertex cover is in fact a dominating set). The \emph{vertex cover number} of~$G$, denoted by $\tau$, is the size of a minimum vertex cover of~$G$. Deciding whether an $n$-vertex graph $G$ admits vertex cover of size $\tau$, and if so computing one, can be done in $O(2^\tau+\tau\cdot n)$ time~\cite{DBLP:journals/tcs/ChenKX10}. Let $G=(V,E)$ be an $n$-vertex DAG with vertex cover number $\tau$. Let $C=\{c_1,c_2,\dots,c_\tau\}$ be a vertex cover of $G$ such that $|C| = \tau$. The next lemma matches an analogous~result~in~\cite{DBLP:journals/jgaa/BhoreGMN20}.

\begin{restatable}[$\star$]{lemma}{lemVCUpperBound}\label{le:vc-upperbound}$G$ admits a $\tau$-page \ube that can be computed in $O(\tau \cdot n)$ time.
\end{restatable}
\newcommand{\lemVCUpperBoundProof}{\begin{proof}
Let $\pi$ be any topological ordering of $G$, that is, any ordering such that for every edge $uv$ of $G$, $\pi(u)<\pi(v)$. It is well-known that any DAG with $n$ vertices and $m$ edges admits at least one topological ordering, which can be computed in $O(n+m)$ time (see, e.g.,~\cite{DBLP:books/daglib/0023376}). Since $G$ has $O(\tau \cdot n)$ edges, $\pi$ can be computed in $O(\tau \cdot n)$ time.  Now consider the following page assignment $\sigma$, which again can be computed in $O(\tau \cdot n)$ time. Let $U=V\setminus C$; for each $1 \le i \le \tau$, we set $\sigma(e)=i$ for all edges $e=(u, c_i)$ with $u \in U\cup \{c_1,\dots,c_{i-1}\}$. By construction, for each $1 \le i \le \tau$, all edges in  page $i$ are incident to $c_i$, and thus no two of them cross each other. Moreover, by definition of topological ordering, for every  edge $uv$ of $G$, it holds $\pi(u)<\pi(v)$. Therefore, the pair $\langle \pi, \sigma \rangle$ is a $\tau$-page \ube of $G$ and has been computed in $O(\tau \cdot n)$ time.\end{proof}}

For a fixed $k \in \mathbb{N}$, if $k \ge \tau$, then $G$ admits a $k$-page \ube by \cref{le:vc-upperbound}. Thus we assume $k < \tau$. 
Two vertices $u,v \in V \setminus C$ are of the same \emph{type} $U$ if they have the same set of neighbors $U \subseteq C$ and, for every $w \in U$, the edges connecting $w$ to $u$ and $w$ to $v$ have the same orientation.  We proceed with the following reduction rule. For each type $U$, let $V_U$ denote the set of vertices of type $U$. 

\smallskip\noindent\textbf{R.1:} {\em If there exists a type $U$ such that $|V_U| \ge 2 \cdot k^\tau +2$, then pick an arbitrary vertex $u \in V_U$ and set $G \coloneqq G - u$.}

\smallskip\noindent  Since there are $2^{\tau}$ different neighborhoods of size at most $\tau$, and for each of them there are at most $2^\tau$ possible orientations, the type relation yields at most $2^{2\tau}$ distinct types. Therefore assigning a type to each vertex and applying \textbf{R.1} exhaustively can be done in $2^{O(\tau)} + \tau \cdot n$ time. We can prove that the~rule~is~safe.

\begin{restatable}[$\star$]{lemma}{lemSafe}\label{lem:safe}
The reduction rule \textbf{R.1} is safe.
\end{restatable}
\begin{proof}[Sketch]
Let $u \in V_U$, such that $|V_U| \ge 2 \cdot k^\tau +2$. Suppose that $G - u$ admits a $k$-page \ube $\langle \pi, \sigma \rangle$. Two vertices $u_1, u_2 \in V_U \setminus \{u\}$ are \emph{page equivalent}, if for each vertex $w \in U$, the edges $u_1w$ and $u_2w$ are both assigned to the same page according to $\sigma$. By definition of type, each vertex in $V_U$ has degree exactly $|U|$, hence this relation partitions the vertices of $V_U$ into at most 
$k^{|U|} \le k^\tau$ sets. Since $|V_U  \setminus \{u\}| \ge 2 \cdot k^\tau + 1$, at least three vertices of this set are page equivalent. One can prove that these three vertices are incident to only one vertex in $C$. Then we can extend $\pi$ by introducing $u$ right next to any of these three vertices, say $u_1$,  and assign each edge $uw$ incident to $u$ to the same page as $u_1w$. 
\end{proof}
\newcommand{\lemSafeProof}{\begin{proof}
Since removing a vertex $u$ from a $k$-page \ube yields a $h$-page \ube of $G - u$ with $h \le k$, one direction follows. 

For the other direction, let $u \in V_U$, such that $|V_U| \ge 2 \cdot k^\tau +2$. Suppose that $G - u$ admits a $k$-page \ube $\langle \pi, \sigma \rangle$. Two vertices $u_1, u_2 \in V_U \setminus \{u\}$ are \emph{page equivalent}, if for each vertex $w \in U$, the edges $u_1w$ and $u_2w$ are both assigned to the same page according to $\sigma$. By definition of type, each vertex in $V_U$ has degree exactly $|U|$, hence this relation partitions the vertices of $V_U$ into at most 
$k^{|U|} \le k^\tau$ sets. Since $|V_U  \setminus \{u\}| \ge 2 \cdot k^\tau + 1$, at least three vertices of this set, which we denote by $u_1$, $u_2$, and $u_3$,  are page equivalent. Consider now the graph induced by the edges of these three vertices that are assigned to a particular page. Since $u_1,u_2,u_3$ are all incident to the same set of $c \ge 1$ vertices in $C$, such a graph is $K_{c,3}$. However, since $K_{2,3}$ is not outerplanar (and hence has no $1$-page \ube), we have that $c=1$. Then we can extend $\pi$ by introducing $u$ right next (or equivalently right before) to $u_1$ and assign each edge $uw$ (or $wu$) incident to $u$ to the same page as $u_1w$ ($wu_1$). Consider any edge $uw$ (or symmetrically $wu$) introduced as described above. Such edge follows the curve of $u_1w$ and hence, since $u_1w$ is uncrossed in its page, $uw$ does not cross any other edge of the same page. Moreover, since $u$ and $u_1$ are of the same type, the edges $uw$ and $u_1w$ are oriented consistently (either both outgoing from or incoming to $w$). Therefore, we obtained a $k$-page \ube of $G$, which implies that \textbf{R.1} is safe.\end{proof}}

\begin{restatable}[$\star$]{theorem}{thmFPTVertexCoverAndThickness}\label{thm:fpt-vertex-cover-and-thickness}
The upward book thickness problem parameterized by the vertex cover number $\tau$ admits a kernel of size $k^{O(\tau)}$. 
\end{restatable}
\newcommand{\thmFPTVertexCoverAndThicknessProof}{\begin{proof}Let $G^*$ be the graph obtained by applying \textbf{R.1} exhaustively.  We have already seen that there are at most $2^{2\tau}$ distinct types. After the application of \textbf{R.1}, each type $U$ is such that  $|V_U| \le 2 \cdot k^\tau+1$ vertices, and therefore $G^*$ contains at most $(2 \cdot k^\tau+1) \cdot 2^{2\tau}+\tau \le k^{O(\tau)}$ vertices and at most $\tau \cdot  k^{O(\tau)}$ edges.\end{proof}}

\begin{restatable}[$\star$]{corollary}{corVCUbe}\label{cor-vc-ube}
Let $G$ be an $n$-vertex graph with vertex cover number $\tau$. For any $k \in \mathbb{N}$, we can decide whether $\ubt{G} \le k$  in $O(\tau^{\tau^{O(\tau)}} + \tau \cdot n)$ time. Also, within the same time complexity, we can compute a $k$-page \ube of $G$, if it exists.
\end{restatable}
\newcommand{\corVCUbeProof}{\begin{proof}
If $k \ge \tau$, by \cref{le:vc-upperbound}, we can immediately return a positive answer and compute a $k$-page \ube in $O(\tau \cdot n)$ time. So assume $k < \tau$. By~\cite{DBLP:journals/tcs/ChenKX10}, we compute a vertex cover $C$ of $G$ of size $\tau$ in $O(2^\tau + \tau \cdot n)$ time. We have seen that computing the kernel $G^*$ of $G$ can be done in $O(2^{O(\tau)} + \tau \cdot n)$ time. Since we have $2^{2\tau}$ types, and each of the $2 \cdot k^\tau + 2$ elements of the same type are equivalent in the book embedding (the position of two elements of the same type can be exchanged in a linear order without affecting the page assignment), the number of linear orders is $(2^{2\tau})^{O(k^\tau)}=2^{{O(\tau^{\tau+1})}}$. 
Since $G^*$ contains $\tau \cdot k^{O(\tau)}=\tau^{O(\tau)}$ edges, the number of page assignments is then $k^{\tau^{O(\tau)}}$. Overall, the book thickness of $G^*$ can be determined in $2^{{O(\tau^{\tau+1})}} \cdot k^{\tau^{O(\tau)}}=k^{\tau^{O(\tau)}}=\tau^{\tau^{O(\tau)}}$ time. If $G^*$ admits a $k$-page \ube, in $O(\tau \cdot n)$ time, we can reinsert the vertices in $G \setminus G^*$ as in the proof of \cref{lem:safe}, thus obtaining a $k$-page \ube of $G$ in overall $O(\tau^{\tau^{O(\tau)}}{+}\tau \cdot n)$ time.\end{proof}}

\section{Open Problems}
The next questions naturally arise from our research: (i)~Is the UBT of upward outerplanar graphs bounded by a constant? (ii)~Are there other parameters that are larger than the domination number (and possibly smaller than the vertex cover number) for which the problem is in \textsf{FPT}? (iii)~Does the upward book thickness problem parameterized by vertex cover number~admit~a~polynomial~kernel?

\bibliographystyle{splncs04}
\bibliography{bibliography}
\clearpage
\appendix
\section{Missing Proofs of \cref{ref:prelims}}

\lembimodality*
\lembimodalityproof

\section{Missing Proofs of \cref{se:outerplanar}}

\subsection{Missing Proofs of \cref{sse:max-outerpaths}}

\fandecunique*
\fandecuniqueProof

\lemSharededge*
\lemSharededgeProof

\lemOneIsertion*
\lemOneIsertionProof

\lemStouterpath*
\lemStouterpathProof

\lemStOuterExsists*
\lemStOuterExsistsProof

\lemConsecutiveshare*
\lemConsecutiveshareProof

\lemConsecutiveFans*
\lemConsecutiveFansProof

\lemConsecutiveV*
\lemConsecutiveVProof

    \thmMain*
\thmMainProof

\subsection{Missing Proofs of \cref{sse:outerplanar}}

\lemStOuterplanar*
\lemStOuterplanarProof

\thmGeneral*
\thmGeneralProof

\subsection{Missing Proofs of \cref{sse:cactus}}

\lemNoTrivial*
\lemNoTrivialProof

\lemCycle*
\lemCycleProof

\thmMainCactus*
\thmMainCactusProof

\section{Missing Proofs of \cref{se:complexity}}

\subsection{Missing Proofs of \cref{sse:hardness}}

\lemReduction*
\lemReductionProof

\subsection{Missing Proofs of \cref{sse:fpt}}

\lemVCUpperBound*
\lemVCUpperBoundProof

\lemSafe*
\lemSafeProof

\thmFPTVertexCoverAndThickness*
\thmFPTVertexCoverAndThicknessProof

\corVCUbe*
\corVCUbeProof
\end{document}